\documentclass[12pt]{amsart}
\usepackage{amsfonts,amssymb,amsthm}

\advance\textheight\topskip
\allowdisplaybreaks

\usepackage[curve,matrix,arrow,frame,ps,line]{xy} 
\UsePSspecials{dvips}

\usepackage{upgreek}
\usepackage[mathscr]{eucal}

\usepackage{times}
\usepackage{mathptmx}

\usepackage[bookmarks=false, hyperfigures=false, a4paper]{hyperref}

\newcommand{\pT}{t}

\newcommand{\nsize}[1]{\mbox{\normalsize$\displaystyle#1$}}

\newcommand{\Mat}{\mathrm{Mat}}

\newcommand{\ribbon}{{\pmb{v}}}

\newcommand{\coint}{{\pmb{\upLambda}}}

\newcommand{\diff}{d}

\newcommand{\Dz}{\dd}
\newcommand{\dz}{\zeta}
\newcommand{\dDz}{\delta}

\newcommand{\Czd}{\oC_{\q}[z,\Dz]}
\newcommand{\bCzd}{\overline{\oC}_{\q}[z,\Dz]}
\newcommand{\OCzd}{\Omega\overline{\oC}_{\q}[z,\Dz]}

\newcommand{\algH}{\mathscr{H}}

\newcommand{\II}{\mathscr{I}}
\newcommand{\JJ}{\mathscr{J}}
\newcommand{\TT}{\mathscr{T}}

\newcommand{\rep}{\mathscr}
\newcommand{\repX}{\rep{X}}

\newcommand{\repP}{\rep{P}}

\newcommand{\bref}[1]{\textbf{\ref{#1}}}

\newcommand{\SL}[1]{s\ell(#1)}

\renewcommand{\geq}{\geqslant}
\renewcommand{\leq}{\leqslant}

\newcommand{\tensor}{\otimes}

\newcommand{\q}{\mathfrak{q}}

\newcommand{\floor}[1]{\lfloor#1\rfloor}

\newcommand{\UresSL}[1]{\overline{\mathscr{U}}_{\q} s\ell(#1)}
\newcommand{\U}{\mathscr{U}}

\newcommand{\mfrac}[2]{\raisebox{.8pt}{\mbox{\small$\displaystyle\frac{#1}{#2}$}}}
\newcommand{\ffrac}[2]{\raisebox{.5pt}{\mbox{\footnotesize$\displaystyle\frac{#1}{#2}$}}}
\newcommand{\fffrac}[2]{\raisebox{.9pt}{\mbox{\scriptsize$\displaystyle\frac{#1}{#2}$}}}
\newcommand{\half}{%
  \mathchoice{\ffrac{1}{2}}{\frac{1}{2}}{\frac{1}{2}}{\frac{1}{2}}}

\newcommand{\qfac}[1]{[#1]!\,}
\newcommand{\qint}[1]{[#1]}

\newcommand{\qbin}[2]{\mathchoice%
  {\qbinm{#1}{#2}}{\qbinmm{#1}{#2}}%
  {\qbinmm{#1}{#2}}{\qbinmm{#1}{#2}}}
\newcommand{\qbinm}[2]{\mbox{\footnotesize$\displaystyle
    \genfrac{[}{]}{0pt}{}{#1}{#2}$}}
\newcommand{\qbinmm}[2]{\genfrac{[}{]}{0pt}{}{#1}{#2}}

\newcommand{\nbin}[2]{\mathchoice%
  {\nbinm{#1}{#2}}{\nbinmm{#1}{#2}}%
  {\nbinmm{#1}{#2}}{\nbinmm{#1}{#2}}}
\newcommand{\nbinm}[2]{\mbox{\footnotesize$\displaystyle
    \genfrac{(}{)}{0pt}{}{#1}{#2}$}}
\newcommand{\nbinmm}[2]{\genfrac{(}{)}{0pt}{}{#1}{#2}}

\newcommand{\dd}{\partial}

\newcommand{\oC}{\mathbb{C}}

\newcommand{\oZ}{\mathbb{Z}}

\numberwithin{equation}{section}
\makeatletter
\@addtoreset{equation}{section}
\@addtoreset{subsubsection}{section}

\def\@secnumfont{\bfseries}
\def\subsubsection{\@startsection{subsubsection}{3}%
  \z@{.5\linespacing\@plus.7\linespacing}{-.5em}%
  {\normalfont\bfseries}}
\def\paragraph{\@startsection{paragraph}{4}%
  \z@\z@{-\fontdimen2\font}%
  \normalfont\bfseries}
\def\subparagraph{\@startsection{subparagraph}{5}%
  \z@\z@{-\fontdimen2\font}%
  \normalfont\bfseries}

\makeatother

\swapnumbers

\newtheorem{thm}[subsubsection]{Theorem}

\newtheorem{lemma}[subsubsection]{Lemma}

\theoremstyle{definition}

\begin{document}

\title[Differential $\U$-module algebra]{A differential $\U$-module
  algebra for $\U=\UresSL2$ at an even root of unity}

\author[Semikhatov]{A.M.~Semikhatov}%

\address{Lebedev Physics Institute
  \hfill\mbox{}\linebreak \texttt{ams@sci.lebedev.ru}}

\begin{abstract}
  We show that the full matrix algebra $\Mat_p(\oC)$ is a $\U$-module
  algebra for $\U=\UresSL2$, a $2p^3$-dimensional quantum $\SL2$ group
  at the $2p$th root of unity.  $\Mat_p(\oC)$ decomposes into a direct
  sum of projective $\U$-modules $\repP^+_n$ with all odd~$n$, $1\leq
  n\leq p$.  In terms of generators and relations, this $\U$-module
  algebra is described as the algebra of $q$-differential operators
  ``in one variable'' with the relations $\Dz\, z = \q - \q^{-1} +
  \q^{-2} z\,\Dz$ and $z^p=\Dz^p=0$.  These relations define a
  ``parafermionic'' statistics that generalizes the fermionic
  commutation relations.  By the Kazhdan--Lusztig duality, it is to be
  realized in a manifestly quantum-group-symmetric description of
  $(p,1)$ logarithmic conformal field models.  We extend the
  Kazhdan--Lusztig duality between $\U$ and the $(p,1)$ logarithmic
  models by constructing a quantum de~Rham complex of the new
  $\U$-module algebra and discussing its field-theory counterpart.
\end{abstract}

\maketitle

\thispagestyle{empty}


\section{Introduction}
\subsection{The main results}
For an integer $p\geq2$, let $\q=e^{\frac{i\pi}{p}}$ and let
$\U=\UresSL2$ be the quantum group with generators $E$, $K$, and $F$
and the relations
\begin{gather}\label{the-qugr}
  \begin{aligned}
    KEK^{-1}&=\q^2E,\quad
    KFK^{-1}=\q^{-2}F,\\
    [E,F]&=\ffrac{K-K^{-1}}{\q-\q^{-1}},
  \end{aligned}
  \\
  \label{the-constraints}
  E^{p}=F^{p}=0,\quad K^{2p}=1
\end{gather}
(and the Hopf algebra structure to be described below).

We construct a representation of $\U$ on the full matrix algebra
$\Mat_p(\oC)$.  For a $p\times p$ matrix
$X=(x^{\vphantom{c}}_{ij})$, \ $(EX)^{\vphantom{c}}_{ij}$ is a linear
combination of the right and upper neighbors of
$x^{\vphantom{c}}_{ij}$, and $(FX)^{\vphantom{c}}_{ij}$ is a linear
combination of the left and lower neighbors, with the coefficients
shown in the diagrams:
\begin{equation}\label{diagrams}
  E:
  \begin{array}{c}
    \xymatrix@C=50pt@R=40pt{
      *+[F-]{i-1,j}
      \ar_{\textstyle-\fffrac{\q^{2(i-j-1)}}{\q-\q^{-1}}}[d]
      &\\
      *+[F-,]{i,j}&*+[F-]{i,j+1}\ar_(.53){(\q-\q^{-1})^{-1}}[l]
    }
  \end{array}
  \
  F:
  \begin{array}{c}
    \xymatrix@C=60pt@R=40pt{
      *+[F-]{i,j-1}\ar^(.56){-\q^{j-2i}\qint{j-1}}[r]
      &*+[F-,]{i,j}\\
      &*+[F-]{i+1,j}\ar_{\q^{1-i}\qint{i}}[u]
    }
  \end{array}
\end{equation}

\bigskip

\noindent
With the necessary modifications at the boundaries, the precise
formulas are as follows:\pagebreak[3]
\begin{small}%
\begin{align}
  \nsize{E(X)}
  &\nsize{{}={}}
  \ffrac{1}{\q - \q^{-1}}\renewcommand{\arraycolsep}{0.5pt}
  \begin{pmatrix}
    x^{\vphantom{c}}_{12}&\dots&x^{\vphantom{c}}_{i,j+1}&\dots&0
    \\
    \vdots&\ddots&\vdots&   &\vdots
    \\
    x^{\vphantom{c}}_{i,2}-\q^{2 (i-2)}x^{\vphantom{c}}_{i-1,1}
    &\dots&x^{\vphantom{c}}_{i,j+1}-\q^{2 (i-j-1)}x^{\vphantom{c}}_{i-1,j}&\dots
    &-\q^{2 (i-1)} x^{\vphantom{c}}_{i-1,p}
    \\
    \vdots&   &\vdots&\ddots&\vdots
    \\
    x^{\vphantom{c}}_{p,2}-\q^{-4} x^{\vphantom{c}}_{p-1,1}
    &\;\dots&x^{\vphantom{c}}_{p,j+1}-\q^{-2 (j+1)}x^{\vphantom{c}}_{p-1,j}&\;\dots&
    -\q^{-2} x^{\vphantom{c}}_{p-1,p}
  \end{pmatrix}
  \notag
  \\
  \intertext{\normalsize(with the only zero in the top right corner),
    \textit{where we explicitly show the $i$th row and the $j$th
      column};}
  \nsize{(KX)^{\vphantom{c}}_{ij}}&
  \nsize{{}=\q^{2(i-j)}x^{\vphantom{c}}_{ij};}\notag
  \\
  \intertext{\normalsize and}
  \nsize{F(X)}
  &\nsize{{}={}}\renewcommand{\arraycolsep}{0.5pt}
  \begin{pmatrix}
    x^{\vphantom{c}}_{21}&\dots
    &x^{\vphantom{c}}_{2,j}-\q^{j - 2}\qint{j - 1}x^{\vphantom{c}}_{1,j-1}&\dots
    &x^{\vphantom{c}}_{2,p}+\q^{-2}x^{\vphantom{c}}_{1,p-1}
    \\
    \vdots&\ddots&\vdots&   &\vdots
    \\
    \q^{1 - i}\qint{i}x^{\vphantom{c}}_{i+1,1}\ &\dots
    &
    \q^{1 - i}\qint{i}x^{\vphantom{c}}_{i+1,j}
    - \q^{j - 2i}\qint{j - 1}x^{\vphantom{c}}_{i,j-1}\ &\dots
    &\q^{1 - i}\qint{i}x^{\vphantom{c}}_{i+1,p}+\q^{-2i}x^{\vphantom{c}}_{i,p-1}
    \\
    \vdots&   &\vdots&\ddots&\vdots
    \\
    0&\;\dots&-\q^j \qint{j-1}x^{\vphantom{c}}_{p,j-1}&\;\dots&
    x^{\vphantom{c}}_{p,p-1}
  \end{pmatrix}
  \notag
\end{align}%
\end{small}%
(with the only zero in the bottom left corner), where we again show
the $i$th row and the $j$th column, and where
\begin{equation*}
  \qint{n}=\ffrac{\q^n - \q^{-n}}{\q - \q^{-1}}.
\end{equation*}

\medskip
\noindent\textbf{Theorem.}
\begin{it}
  \begin{enumerate}
  \item The above formulas define a representation of \ $\U=\UresSL2$
    on $\Mat_p(\oC)$.

  \item $\Mat_p(\oC)$ is a $\U$-module algebra.

  \end{enumerate}
\end{it}
We recall that for a Hopf algebra $\algH$, an $\algH$-module algebra
is an algebra in the tensor category of $\algH$-modules, i.e., is a
(left) $\algH$-module $V$ with a composition law $V\tensor V\to V$
such that $h(v\,w)=\sum h'(v)\,h''(w)$ for $h\in\algH$ and $v,w\in V$
(here, $\Delta(h)=\sum h'\tensor h''$ is Sweedler's notation for
coproduct).\footnote{In simple words, the condition states a natural
  compatibility between the $\algH$-action and multiplication on $V$,
  ``natural'' because $\algH$ acts on a product via comultiplication.
  Claim~(2) is thus that the standard matrix multiplication is
  compatible with the proposed action of~$\U$ (and its
  comultiplication).}

The quantum group $\U$ has $2p$ irreducible representations
$\repX^\pm_r$, $1\leq r\leq p$, with
$\dim\repX^\pm_r=r$~\cite{[FGST]}.  We let $\repP^\pm_r$ denote their
projective covers.  The ``plus'' representations are distinguished
from the ``minus'' ones by the fact that tensor products
$\repX^+_r\tensor\repX^+_s$ decompose into the $\repX^+_{r'}$ and
$\repP^+_{r'}$ (and $\repX^+_1$ is the trivial representation).

\medskip
\noindent\textbf{Theorem} (continued)\textbf{.}
\begin{it}
  \begin{enumerate}\setcounter{enumi}{2}
  \item $\Mat_p(\oC)$ \textit{decomposes} into indecomposable
    $\U$-modules as
    \begin{equation}\label{the-decomposition}
      \Mat_p(\oC)=\repP^+_1\oplus\repP^+_3\oplus\dots\oplus\repP^+_\nu,
      \pagebreak[3]
    \end{equation}\pagebreak[3]%
    where $\nu=p$ is $p$ if odd and $p-1$ if $p$ is even.
  \end{enumerate}
\end{it}

The algebra in~\eqref{the-decomposition} is the smallest $\U$-module
algebra that contains $\repP^+_1$, the projective cover of the trivial
representation.  This $2p$-dimensional module can be visualized as a
span of $2p$ elements with the $\U$-action given by~\cite{[FGST]}
\begin{equation*}
  \xymatrix@=12pt@C=6pt{
    &&&&&\pT
    \ar^{F}[dr]
    \ar_{E}[dl]
    \\  
    \ell_{p-1}
    &\kern-10pt\rightleftarrows\kern-10pt
    &
    \ell_{p-2}
    &
    \kern-10pt\rightleftarrows \ldots \rightleftarrows\kern-10pt
    &\ell_1
    \ar_{F}[dr]
    &{}
    &r_1
    \ar^{E}[dl]
    &\kern-10pt\rightleftarrows \ldots \rightleftarrows\kern-10pt
    &
    r_{p-2}
    &\kern-10pt\rightleftarrows\kern-10pt
    &
    r_{p-1}
    \\
    &&&&&
    1
  }
\end{equation*}
where the horizontal arrows represent the action of $E$ (to the left)
and $F$ (to the right) up to \textit{nonzero} factors and the tilted
arrows indicate that the map in the opposite direction vanishes; the
bottom $1$ spans the $1$-dimensional submodule.  In the
\textit{algebra} defined on the sum of projective modules, we can say
more about the structure of~$\repP^+_1$.

\medskip
\noindent\textbf{Theorem} (continued)\textbf{.}
\begin{it}
  \begin{enumerate}\setcounter{enumi}{3}
  \item There is an isomorphism of \ $\U$-module algebras
    \begin{equation*}
      \repP^+_1\oplus\repP^+_3\oplus\dots\oplus\repP^+_\nu\cong\bCzd,
    \end{equation*}
    where $\bCzd$ is the associative algebra with generators $\Dz$ and
    $z$ and the relations
    \begin{gather}\label{the-relation}
      \Dz\, z = \q - \q^{-1} + \q^{-2} z\,\Dz,
      \\
      \label{^p-relations}
      \Dz^p=0,\qquad z^p=0.
    \end{gather}

  \item Under this isomorphism, the ``wings'' of the projective module
    $\repP^+_1$ in~\eqref{the-decomposition} are powers of a single
    generator each,
    \begin{equation}\label{repPp1}
      \xymatrix@=12pt@C=6pt{
        &&&&&\pT
        \ar^{F}[dr]
        \ar_{E}[dl]
        \\  
        z^{p-1}
        &\kern-10pt\rightleftarrows\kern-10pt
        &
        z^{p-2}
        &
        \kern-10pt\rightleftarrows \ldots \rightleftarrows\kern-10pt
        &z
        \ar_{F}[dr]
        &{}
        &\Dz
        \ar^{E}[dl]
        &\kern-10pt\rightleftarrows \ldots \rightleftarrows\kern-10pt
        &
        \Dz^{p-2}
        &\kern-10pt\rightleftarrows\kern-10pt
        &
        \Dz^{p-1}
        \\
        &&&&&
        1
      }
    \end{equation}
    and the ``top'' element is
    \begin{equation}\label{the-t}
      t=\sum_{i=1}^{p-1}\ffrac{1}{\qint{i}}\,z^i\,\Dz^i.
    \end{equation}
  \end{enumerate}
\end{it}
In other words, our $\U$-module algebra is identified with the algebra
of $q$-differential operators ``in one variable'' with nilpotency
conditions~\eqref{^p-relations} (and with a slightly unusual rule for
carrying $\Dz$ through~$z$).  This is parallel to the classic result
that $\Mat_p(\oC)$ is generated by $x$ and $y$ satisfying the
relations $yx = q xy$ and $x^p = y^p = 1$, where $q$ is the $p$th root
of unity~\cite{[W]} (a finite quantum plane in the modern
terminology), but there seems to be no direct (``exponential'')
relation between our ``nilpotent'' ($\Dz^p=z^p=0$) and the classic
``unipotent'' ($x^p=y^p=1$) constructions.  Apart from matrix
curiosities, the $q$-differential operators yield a preferential
(``more invariant'') description of the algebra on the sum of ``odd''
projective $\U$-modules $\repP^+_1\oplus\repP^+_3\oplus\ldots$
compared with its matrix realization.

Obviously, $t$ in~\eqref{repPp1} is defined up to the addition of
$\alpha 1$, $\alpha\in\oC$, and expression~\eqref{the-t} is therefore
a particular representative of this class; this is to be understood in
what follows.

For a quasitriangular $\algH$, an $\algH$-module algebra is said to be
quantum commutative~\cite{[CW]} (also, $\algH$-, $R$-, or braided
commutative) if
\begin{gather}\label{equivariance}
  v\,w=\sum R^{(2)}(w)\,R^{(1)}(v),
\end{gather}
for all $v,w\in V$, where $R=\sum R^{(1)}\tensor R^{(2)} \in
\algH\tensor\algH$ is the universal $R$-matrix.  Our $\U$-module
algebra is not quantum commutative; nevertheless,
\textit{relation~\eqref{equivariance} is satisfied for $v=z^i\Dz^j$
  and $w=z^{m} \Dz^{n}$ if and only if either $n=0$ or $i=0$ or $|i +
  m - j - n|\geq p$
}.

Returning to matrices and representing commutation
relations~\eqref{the-relation} as\footnote{We do not reduce the
  expressions using that $\q^{p}=-1$ and $\qint{p-i}=\qint{i}$ when
  the unreduced form helps to see a pattern.}
\begin{footnotesize}%
  \begin{gather}
    \label{zDz-matrices}
    \nsize{z={}}
    \begin{pmatrix}
      0&\hdotsfor{3}&0\\
      1&0&\hdotsfor{2}&0\\
      0&1&0&\dots&0\\
      \vdots&&\ddots&\ddots&\vdots\\
      0&\hdotsfor{2}&1&0
    \end{pmatrix},\qquad
    \nsize{\Dz=(\q-\q^{-1})}
    \begin{pmatrix}
      0 & 1 & \hdotsfor{2} & \kern-12pt0\\
      0 & 0 & \;\q^{-1}\qint{2}\!\! & \hdots & \kern-12pt0\\
      \vdots&  & \ddots & \ddots & \kern-9pt\vdots \\
      0 & \hdotsfor{2}  &   0 & \q^{2-p}\qint{p-1}\\
      0 & \hdotsfor{2}  & 0 & \kern-12pt 0
    \end{pmatrix},
    \\
    \intertext{\normalsize we have one of the ``matrix
      curiosities''---\,\textit{integers rather than $q$-integers} in
      the matrix representation of~\eqref{the-t}:} \nsize{t=
      (\q-\q^{-1})}
    \begin{pmatrix}
      0&0&\hdotsfor{3}&\kern-6pt 0\\
      0&1&0&\hdotsfor{2}&\kern-6pt 0\\
      0&0&2&0&\hdotsfor{1}&\kern-6pt 0\\
      \vdots&&&\ddots&&\kern-6pt \vdots\\
      0&\hdotsfor{2}&0&p-2&\kern-6pt 0\\
      0&\hdotsfor{3}&0&\kern-6pt p-1
    \end{pmatrix}.
  \end{gather}%
\end{footnotesize}%

Next, it turns out that a differential calculus can be developed for
our algebra $\bCzd$ such that the differential (satisfying the
``classical'' Leibnitz rule!) commutes with the quantum group action.
Let $\oC_{\q}[\dz,\dDz]$ be an ``odd'' counterpart of
$\bCzd$\,---\,the algebra on $\dz$ and $\dDz$ with the relations
$\dz^2=0$, $\dDz^2=0$, and $\dDz\,\dz=-\q^{-2}\,\dz\,\dDz$.  The new
variables are to be considered the differentials of the
``coordinates,'' $\dz=d(z)$ and $\dDz=d(\Dz)$.\footnote{If our $\Czd$
  is relabeled as $\oC^{2|0}_{\q}[z,\Dz]$, then its ``odd''
  counterpart is to be denoted as~$\oC^{0|2}_{\q}[\dz,\dDz]$; we use
  the simpler notation for brevity.}

\medskip
\noindent\textbf{Theorem} (continued)\textbf{.}
\begin{it}
  \begin{enumerate}\setcounter{enumi}{5}
  \item A quotient of $\;\bCzd\tensor\oC_{\q}[\dz,\dDz]$ can
    be\pagebreak[3] endowed with the structure of a differential
    $\U$-module algebra $(\OCzd, d)$ that is a quantum de~Rham complex
    of~$\;\bCzd$.
  \end{enumerate}
\end{it}
The notation $\OCzd$ assumes that $\bCzd$ is the algebra of
$\;0$-forms.  The exact formulas defining the quotient and the
$\U$-action are given in Sec.~\ref{sec:calculus} below.  

As an illustration of the action of the differential $d$ on the module
structure, we note that the unity in $\repP^+_1$, Eq.~\eqref{repPp1},
is annihilated, and therefore $\repP^+_1$ is not preserved by~$d$.  On
the other hand, there is another, not $d$-closed element in
$\Omega^1\bCzd$ in the same grade as~$dt$, \textit{and} elements in
the cohomology of~$d$, which together with $d(\repP^+_1)$ arrange into
the direct sum of two $\U$-modules
\begin{multline}\label{two-modules}
  \xymatrix@=12pt@C=5pt{
    z^{p-1}\,\dz
    \ar^F[dr]&&&&&&
      \displaystyle\smash[t]{\sum_{i=1}^{p - 1}}
      \fffrac{1}{\qint{i}}d(z^{i})\,\Dz^{i}
    \ar_(.75)E[dl]
    \\
    &z^{p-2}\,\dz
    &\kern-6pt\rightleftarrows \ldots \rightleftarrows\kern-6pt
    &z\,\dz
    &\kern-10pt\rightleftarrows\kern-10pt
    &\dz&
  }
  \\
  \mbox{\Large$\oplus$}
  \\
  \xymatrix@=12pt@C=5pt{
    \displaystyle\smash[t]{\sum_{i=1}^{p - 1}}
    \fffrac{1}{\qint{i}}z^{i}\,d(\Dz^{i})
    \ar^(.6)F[dr]
    &&&&&&\Dz^{p-1}\,\dDz\ar_E[dl]
    \\
    &\dDz
    &\kern-10pt\rightleftarrows\kern-10pt
    &\Dz\,\dDz
    &\kern-6pt\rightleftarrows \ldots \rightleftarrows\kern-6pt
    &\Dz^{p-2}\,\dDz
  }
\end{multline}
where $d(z^{i}) = \q^{1-i}\qint{i}\,z^{i-1}\,\dz$ and $d(\Dz^{i}) =
\q^{i-1}\qint{i}\,\Dz^{i-1}\,\dDz$.  \textit{The ``corners''
  $z^{p-1}\,\dz$ and $\Dz^{p-1}\,\dDz$ are in the cohomology of the
  differential}.

\subsection{Motivation and some (un)related approaches}
Our interest in the quantum group $\U=\UresSL2$ and related objects
stems from its occurrence in logarithmic conformal field
theories~\cite{[FGST],[FGST2]} (also see a similar quantum group
structure in~\cite{[FGST3],[FGST-q]}, a review in~\cite{[S-q]}, and a
further development in~\cite{[G]}).\footnote{On the subject of
  logarithmic $(p,1)$ models, without attempting to be complete in any
  way, we note the pioneering works~\cite{[K-first],[GK1],[GK2]}
  (where, in particular, the symmetry of the model\,---\,the
  \textit{triplet algebra}\,---\,was identified),
  reviews~\cite{[G-alg],[F-bits]} of the early stages, ``logarithmic
  deformations'' in~\cite{[FFHST]}, the definition of the triplet
  algebra $W(p)$ at general $p$ as the kernel of a screening and the
  fusion algebra of the $2p$ irreducible
  $W(p)$-representations~\cite{[FHST]} (also see~\cite{[FK]}), the
  study of $W(p)$ with the aid of Zhu's algebra~\cite{[AM-triplet]},
  interesting recent advances in~\cite{[FG],[GR1],[GR2],[G]}, and, of
  course, the numerous references therein.}  But this particular
version of the quantum $\SL2$ actually made its first appearance much
earlier; a regrettable omission in (the arXiv version of)~\cite{[S-q]}
was paper~\cite{[AGL]}, where the regular representation of $\U$ was
elegantly described in terms of the even subalgebra of a matrix
algebra times a Grassmann algebra on two generators for each block
(also see~\cite{[Coq],[DNS],[CGT]} for a very closely related quantum
group at~$p=3$; our quantum group was also the subject of attention
in~\cite{[Ar],[FHT]}).

The correspondence between $\U$ and the $(p,1)$ logarithmic conformal
field models, which is a version of the Kazhdan--Lusztig
duality~\cite{[KL]},\pagebreak[3] extends not only to the
representation theories but also to modular group actions, the modular
group action generated from the characters of the $W(p)$ algebra being
isomorphic to that on the quantum group
center~\cite{[FGST],[FGST2],[FGST3],[FGST-q]}.  But algebraic
structures on $\U$-modules have not been investigated in the
Kazhdan--Lusztig context.  Relations~\eqref{the-relation},
\eqref{^p-relations} are in fact a quantum-group counterpart of the
``hidden'' quantum-group symmetry of the $(p,1)$ logarithmic conformal
model (see~\bref{sec:parafermions} below).

On the other hand, commutation relation~\eqref{the-relation} can be
compared with the (considerably more general) setting of quantum Weyl
algebras~\cite{[WZ],[Zh2],[Zh3]}.  There, one considers the defining
relations (the $\dd^j$ are not powers of an element but different
elements)
\begin{equation*}
\begin{gathered}
  \sum R_{ij}^{kl}x_k x_l = q x_i x_j,\\
  \dd^j x_i = \delta_{i}^{j} + q \sum R^{jk}_{il}x_k \dd^l,\\
  \sum R_{kl}^{ij}\dd^k\dd^l = q \dd^i \dd^j,
\end{gathered}\quad 1\leq i,j,\dots\leq n,
\end{equation*}
where $R$ is an $n^2\times n^2$ matrix solution of the Yang--Baxter
equation \textit{and} the Hecke relation.  For the ``$g\ell_n$''
$R$-matrix, in particular,
\begin{gather*}
  \dd^i x_i = 1 + q^2 x_i \dd^i +(q^2-1)\sum_{j>i}x_j \dd^j,
\end{gather*}
which in the case $n=1$ (of little interest in the general theory of
quantum Weyl algebras) becomes
\begin{gather*}
  \dd x = 1 + q^2 x \dd.
\end{gather*}

Our relation~\eqref{the-relation} involves $\q - \q^{-1}$ instead of
unity, which is dictated by the $\U$-module algebra property, with
$\U=\UresSL2$ being our main, initial object (in contrast to quantum
Weyl algebras, where the ``$\dd\,x$--$x\,\dd$'' relations are
considered primary and then quantum enveloping algebras generated by
the $x_i\dd^j$ are studied; also, our $R$-matrix does not satisfy the
Hecke relation).

\subsection{``Parafermionic statistics''}\label{sec:parafermions}
\subsubsection{}
Relations~\eqref{the-relation} and~\eqref{^p-relations} take a
``fermionic'' form for $p=2$:
\begin{equation*}
   \{\Dz,\,\Dz\}=0,\quad \{z,\,z\}=0,\quad\{\Dz\,,z\}=2i,
\end{equation*}
where $\{~,~\}$ is the anticommutator.\footnote{These three
  anticommutators are not unrelated to, but must be clearly
  distinguished from the relations in the $\U$ algebra itself at
  $p=2$, which can be written as $\{E,\,E\}=0$, $\{\tilde F,\,\tilde
  F\}=0$, and $\{E,\tilde F\}=\frac{1}{2i}(1-K^2)$ for $\tilde F=K
  F$.}  This ``fermionic statistics'' (i.e., Clifford-algebra
relations) is very well known to be relevant to the simplest
logarithmic conformal field theory model in the $(p,1)$ family, the
$(2,1)$ model~\cite{[GK1],[GK2]}, whose dual quantum group is our $\U$
at $p=2$ ($\q=i$): this model is described by ``symplectic
fermions''\,---\,conformal fields defined on the complex plane that
satisfy the fermionic commutation relations~\cite{[K-sy]}.  For
general~$p$, the $(p,1)$ logarithmic model corresponds under the
Kazhdan--Lusztig duality just to $\U$ at $\q=e^{\frac{i\pi}{p}}$, and
relations~\eqref{the-relation} and~\eqref{^p-relations} are a
generalization of the fermionic statistics.

\subsubsection{Manifestly quantum-group-invariant description of LCFTs}
For $p>2$, an important problem is to describe the $(p,1)$ logarithmic
conformal models in \textit{manifestly quantum-group-invariant terms}.
The idea of an explicit quantum group symmetry was (somewhat
implicitly) expressed in~\cite{[FGST2]}, where the Fermi statistics
realized for $p=2$ was predicted to extend for general~$p$ to a
``parafermionic''\footnote{The word ``parafermionic'' is somewhat
  overloaded here (and, in particular, is not related to the
  parafermions discussed in the context of logarithmic conformal field
  theories in~\cite{[S-branching]}); although its use is motivated by
  the discussion in~\cite{[Sm]}, ``anyonic'' might be a better
  choice.} statistics on $p-1$ pairs of fields,
which would also allow realizing projective modules over the triplet
algebra.

Relations~\eqref{the-relation} and~\eqref{^p-relations} suggest this
general-$p$, ``parafermionic'' statistics of the $(p,1)$ logarithmic
conformal field theory models.  To realize it, we introduce $p-1$
pairs of fields $\upzeta^m(w)$ and $\updelta^m(w)$, $m=1,\dots,p-1$,
carrying the same $\U$ representation as the $z^m$ and $\Dz^m$, and
set $\updelta^0(w)=\upzeta^0(w)=1$ (here, $w$ is a coordinate on the
complex plane).  The $\upzeta^m(w)$ and $\updelta^m(w)$ have conformal
weight zero.  With~\eqref{repPp1} rewritten in terms of the fields,
\begin{equation}\label{fieldsPp1}
  \kern-10pt\xymatrix@=12pt@C=6pt{
    &&&\upLambda(w)=\displaystyle
    \smash{\sum_{n=1}^{p-1}}\fffrac{1}{\qint{n}}\upzeta^n\updelta^n(w)
    \ar^(.6){F}[dr]
    \ar_(.6){E}[dl]
    \\  
    \upzeta^{p-1}(w)
    &
    \kern-10pt\rightleftarrows \ldots \rightleftarrows\kern-10pt
    &\upzeta^1(w)
    \ar_{F}[dr]
    &{}
    &\updelta^1(w)
    \ar^{E}[dl]
    &\kern-10pt\rightleftarrows \ldots \rightleftarrows\kern-10pt
    &    
    \updelta^{p-1}(w),
    \\
    &&&
    1
  }
\end{equation}
it follows that $\upLambda(w)$ is a \textit{logarithmic partner} of
the identity operator (cf.~\cite{[FGST2]}).

\subsubsection{First-order ``parafermionic'' systems}
The differential $\diff$ acting on conformal fields (in terms of the
coordinate $w$ on the complex plane) commutes with the quantum group
action on the fields.  This is also the case with $d$ in the de~Rham
complex $\OCzd$ on the algebraic side, and we do not therefore
distinguish the two differentials.  It is instructive to
rewrite~\eqref{two-modules} in terms of fields.  For this, we
introduce the fields $\upeta^n(w)$ as
\begin{equation}\label{eq:eta}
  \diff\updelta^n(w)=\qint{n}\q^{n-1}\upeta^n(w),\quad n=1,\dots,p-1.
\end{equation}
Then the fields $\upzeta^n(w)$ and $\upeta^n(w)$
constitute a $(p-1)$-component ``parafermionic'' version of the
first-order fermionic system.  The field realization of one of the
modules in~\eqref{two-modules} is
\begin{equation}\label{dfieldsPp1}
  \xymatrix@=12pt@C=6pt{
    \JJ(w)\equiv\displaystyle
    \smash{\sum_{n=1}^{p-1}}
    \q^{n-1}\upzeta^n\upeta^n(w)
    \ar^(.65){F}[dr]
    &&&&
    e^{\sqrt{2p}\,\varphi(w)}
    \ar_(.6)E[dl]
    \\  
    {}
    &\upeta^1(w)
    &\kern-10pt\rightleftarrows \ldots \rightleftarrows\kern-10pt
    &    
    \upeta^{p-1}(w),
  }
\end{equation}
where $\varphi(w)$ is introduced as $\diff\varphi(w)=\II(w)$,
\begin{equation}\label{eq:II}
  \II(w)=
  \sum_{n=1}^{p-1}\fffrac{1}{\qint{n}}\diff\upzeta^n\updelta^n(w),
\end{equation}
and $e^{\sqrt{2p}\,\varphi(w)}$ is the ``screening current''---\,a
field on the complex plane such that taking the first-order pole in
the OPE with it defines a screening operator.

In the Appendix, we consider the ``parafermionic'' fields,
generalizing free fermions, in more detail.  The extension from
fermions ($p=2$) to ``parafermions'' (general $p$) is also closely
related to an algebraic pattern that we now recall.

\subsubsection{}
On the algebraic side, just the same ideology of a ``quantum''
generalization of fermionic commutation relations was put forward
in~\cite{[CW]}.  The guiding principle was that of quantum
commutativity, which ``encompasses commutativity of algebras and
superalgebras on one hand and the quantum planes and superplanes on
the other''~\cite{[CW]}.  A number of examples, including the quantum
plane, were considered in that paper.  We also note the related points
in~\cite{[Lu-alg],[BM]}; in particular, a free algebra on the $\xi_i$
with the relations
\begin{equation*}
  \xi_i \xi_j=R_{ij}^{mn}\xi_m\xi_n
\end{equation*}
(where $R_{ij}^{mn}$ is again a matrix solution of the Yang--Baxter
equation) is quantum commutative in the category of Yetter--Drinfeld
modules over the bialgebra obtained from $R$ via the
Faddeev--Reshetikhin--Takhtajan construction, i.e., the free algebra
on the $c^i_j$ with the relations
\begin{equation*}
  R_{mn}^{ij} c_k^n c_l^m = R_{lk}^{mn} c_m^i c_n^j.
\end{equation*}
(A partly reversed logic has also been used to find solutions of the
Yang--Baxter equation from Yetter--Drinfeld (``Yang--Baxter'')
modules~\cite{[LR-alg]}).

For us, as in~\cite{[FGST],[S-q]}, the quantum group $\U$ is not
reconstructed from some $R$-matrix but is given as the primary object
(originally determined by the Kazhdan--Lusztig duality with
logarithmic conformal field theory).  We then define an algebra on
$\Dz$ and $z$ with the crucial commutation relation given
by~\eqref{the-relation}, verify the $\U$-module property, and find the
algebra decomposition.  Alternatively, it could be possible to start
with the appropriate sum of (the ``odd'') projective quantum-group
modules and conclude somehow that it is an associative algebra; from
this perspective, the results in this paper include finding the
generators ($\Dz$ and $z$) and relations (\eqref{the-relation}
and~\eqref{^p-relations}) in this associative algebra.

\subsection{$\UresSL2$} We quote several results about our quantum
group $\U$ in~\eqref{the-qugr}, \eqref{the-constraints} \cite{[FGST]}.

The Hopf algebra structure of $\U$ is given by
\begin{gather*}
  {}\Delta(E)= E\tensor K + 1\tensor E,\quad
  \Delta(K)=K\tensor K,\quad
  \Delta(F)=F\tensor 1 + K^{-1}\tensor F,
  \\
  \epsilon(E)=\epsilon(F)=0,\quad\epsilon(K)=1,
  \\
  S(E)=-E K^{-1},\quad S(K)=K^{-1},\quad S(F)=-K F.
\end{gather*}
Therefore, in particular, the condition for an algebra $V$ carrying a
representation of $\U$ to be a $\U$-module algebra is that
\begin{align*}
  E(v w) &= (E v)(Kw) + v(Ew),\\
  K(v w) &= (Kv)(Kw),\\
  F(v w) &= F(v) w + (K^{-1}v) Fw
\end{align*}
for $v,w\in V$.

For each $1\leq r\leq p-1$, the projective module $\repP^\pm_r$ that
covers the irreducible representation $\repX^\pm_r$ has
dimension~$2p$; for $r=p$, the projective module coincides with the
irreducible representation~\cite{[FGST]}.  The structure of projective
$\U$-modules is made very explicit in~\cite{[FGST]} and all the
indecomposable representations of $\U$ are classified
in~\cite{[FGST2]} (they can also be deduced from a more general
approach in~\cite{[Erd]}).

The universal $R$-matrix for $\U$ was found in~\cite{[FGST]}:
\begin{equation}\label{the-R}
  R =\ffrac{1}{4p}\sum_{i=0}^{p-1}\sum_{\ a,b=0}^{4p-1}
  \ffrac{(\q-\q^{-1})^i}{\qfac{i}}\,
  \q_{\vphantom{-}}^{\frac{i(i-1)}{2} + i(a-b)-\frac{a b}{2}}
  E^i K_{\vphantom{e}}^{\frac{a}{2}}\tensor F^i K_{\vphantom{e}}^{\frac{b}{2}}.
\end{equation}
Strictly speaking, this is not an $R$-matrix \textit{for the quantum
  group $\U$} because of the half-integer powers of $K$ involved here.
This was discussed in detail in~\cite{[FGST]}; an essential point is
that the so-called monodromy matrix $M=R_{21}R$ \ \textit{is} an
element of $\U\tensor\U$; in our present context, a similar effect is
that we do not have to introduce half-integer powers of $\q$ because
all eigenvalues of $K$, which are $\q^n$, occur with even~$n$ here.
Thus, whenever $K$ acts by $\q^{2n}=e^{\frac{2i\pi n}{p}}$, $0\leq
n\leq p-1$, we set $K^{\half}$ to act by $\q^{n}=e^{\frac{i\pi
    n}{p}}$.

\medskip

The $\q$-integers $\qint{n}$ were defined above, and we also use the
standard notation
\begin{equation*}
  \qfac{n}=\qint{1}\qint{2}\dots\qint{n},\quad
  \qbin{m}{n}=\ffrac{\qfac{m}}{\qfac{m-n}\qfac{n}}
\end{equation*}
(with $\qbin{m}{n}=0$ for $m<n$).

Most of the material that relates to proving the theorem is collected
in Sec.~\ref{sec:qline-all}; some remarks about the matrix realization
are in Sec.~\ref{sec:matrices}; the extension to a differential
algebra (the quantum de Rham complex of $\bCzd$) is given in
Sec.~\ref{sec:calculus}.  Implications of the ``para\-fermionic
statistics'' (i.e., of the commutation relations in our $\U$-module
algebra) for conformal field theory are discussed in the Appendix.

\section{$q$-Differential operators on the line at a root of
  unity}\label{sec:qline-all}
We consider the ``quantum line'' $\oC[z]$, i.e., the space of
polynomials in one variable; the ``quantum'' (i.e., noncommutative)
features are to be seen not in the polynomials themselves but in
operators acting on them (and therefore a \textit{quantum} line is a
certain abuse of speech unless it is endowed with some extra
structures).

\subsection{$z\,$, $\Dz\!$, and a $\U$ action}
\subsubsection{}\label{inf-snake}
We define the $\U$ action on $\oC[z]$ as
\begin{alignat*}{2}
  E\,z^m &=-\q^m \qint{m} z^{m+1},
  \\
  K\,\,z^m &= \q^{2m}\,z^m,
  \\
  F\,z^m &= \qint{m} \q^{1-m}\,z^{m-1}.
\end{alignat*}
That this is indeed a $\U$ action is easy to verify.  Clearly, the
unity spans a submodule.  The module structure of $\oC[z]$ is given by
the diagram (an infinite version of the zigzag modules considered
in~\cite{[FGST2]}; see also~\cite{[Erd]})
\begin{equation*}
  \xymatrix@=12pt@C=2pt{
    &
    \ldots
    &
    z^{2p+1}
    \ar^{F}[dr]
    &
    &z^{2p-1}
    \ar_{E}[dl]
    &\kern-6pt\rightleftarrows\kern-6pt
    &\ldots
    &\kern-6pt\rightleftarrows\kern-6pt
    &z^{p+1}
    \ar^{F}[dr]
    &&z^{p-1}
    \ar_{E}[dl]
    &\kern-6pt\rightleftarrows\kern-6pt
    &\ldots
    &\kern-6pt\rightleftarrows\kern-6pt
    &z
    \ar^{F}[dr]
    &
    \\
    \ldots&&&z^{2p}&&&&&&z^p&&&&&&\ 1
    }
\end{equation*}
where the horizontal $\rightleftarrows$ arrows denote the action by
$F$ (to the right) and $E$ (to the left) up to nonzero factors.

\subsubsection{}The formulas above actually make $\oC[z]$ into a
$\U$-module algebra.  The elementary proof of this fact amounts to the
calculation
\begin{multline*}
  \sum E'(z^m)\,E''(z^n)
  = z^m E(z^n) + E(z^m) K(z^n)
  = -\q^n \qint{n} z^m\,z^{n+1} 
  - \q^m \qint{m} z^{m+1} \q^{2n} z^n\\
  =-(\q^n \qint{n} + \q^{m+2n} \qint{m})z^{m+n+1}=
  -\q^{m + n}\qint{m + n}z^{m+n+1} = E(z^{m+n}),
\end{multline*}
and similarly for $F$.

\subsubsection{}\label{dz-relations}
We next introduce a ``dual'' quantum line $\oC[\Dz]$ of polynomials in
a $q$-derivative operator $\Dz$ on $\oC[z]$, and postulate the
commutation relation~\eqref{the-relation}.  A simple exercise in
recursion then leads to the relations
\begin{gather*}
  \Dz^m\,z^n=\sum_{i\geq0}\q^{-(2 m - i) n + i m - \frac{i(i-1)}{2}}
  \qbin{m}{i} \qbin{n}{i} \qfac{i} \left(\q - \q^{-1}\right)^i
  z^{n-i}\Dz^{m-i}
\end{gather*}
(because of the $\q$-binomial coefficients, the range of $i$ is
bounded above by $\min(m,n)$).  Anticipating the result
in~\eqref{repPp1}, we thus have the commutation relations between
elements of the projective module~$\repP^+_1$.

We let $\Czd$ denote the associative algebra generated by $z$ and
$\Dz$ with relation~\eqref{the-relation}.  In the formulas such as
above, $z$ is the operator of multiplication by $z$, and all
expressions like $\Dz^m z^n$ are understood accordingly; as regards
the \textit{action} of $\Dz$ on $\oC[z]$, it is given by the $i=m$
term in the last formula:
\begin{gather*}
  \Dz^m(z^n)=\q^{m (m - n) +\frac{m(m-1)}{2}}
  \qbin{n}{m} \qfac{m} \left(\q - \q^{-1}\right)^m
  z^{n-m}.
\end{gather*}

\subsubsection{}\label{derivations}
It follows from~\bref{dz-relations} that
\begin{align*}
  \Dz^m z &= \q^{-2m} z \Dz^{m} + \q(1-\q^{-2m})\Dz^{m-1}
  \\[-6pt]
  \intertext{and}
  \smash[t]{\Dz\,z^n}
  &={}\smash[t]{\q^{-2n}\,z^n\,\Dz + \q(1 - \q^{-2n})z^{n-1}},
\end{align*}
and hence \textit{$\Dz^p$ and $z^p$ are central in $\Czd$}.

We note that Lusztig's trick of resolving the ambiguities in
$X\mapsto(\Dz^p X - X \Dz^p)/\qint{p}$ and $X\mapsto(z^p X - X
z^p)/\qint{p}$ then yields two \textit{derivations} of $\Czd$:
\begin{align*}
  \mathfrak{d}:{}\
  &\begin{aligned}
    z^n&\mapsto
    \sum_{i=1}^n(-1)^i\q^{i n - \frac{i(i-1)}{2}}
    \ffrac{\qint{n-i+1}\dots\qint{n}}{\qint{i}}\left(\q - \q^{-1}\right)^i
    z^{n-i}\Dz^{p-i},\\
    \Dz^n&\mapsto0
  \end{aligned}
  \\[-6pt]
  \intertext{and}
  \mathfrak{z}:{}\
  &\begin{aligned}
    z^n&\mapsto0,\\[-4pt]
    \Dz^n&\mapsto
    -\sum_{i=1}^n(-1)^i\q^{i n - \frac{i(i-1)}{2}}
    \ffrac{\qint{n-i+1}\dots\qint{n}}{\qint{i}}\left(\q - \q^{-1}\right)^i
    z^{p-i}\Dz^{n-i}.
  \end{aligned}
\end{align*}

\subsubsection{} We next define the $\U$ action on $\oC[\Dz]$ as
\begin{align*}
  E\,\Dz^n &= \q^{1-n}\qint{n}\Dz^{n-1},\\
  K\,\Dz^n &= \q^{-2 n} \Dz^n,\\
  F\,\Dz^n &= -\q^n \qint{n} \Dz^{n+1}.
\end{align*}
Clearly, this \textit{is} a $\U$ action, the unity $1=\Dz^0$ is a
submodule, and this action makes $\oC[\Dz]$ into a $\U$-module
algebra.

\begin{lemma}
  $\Czd$ is a $\U$-module algebra.
\end{lemma}
\noindent
The proof amounts to verifying that $E$ and $F$ preserve the ideal
generated by the left-hand side of~\eqref{the-relation}:
\begin{multline*}
  E(\Dz\, z - (\q - \q^{-1}) - \q^{-2} z\,\Dz)
  =E(\Dz)\, K z + \Dz\, E(z) 
  - \q^{-2}(E(z)\,K(\Dz) + z\,E(\Dz))\\*
  =\q^2 z - \q \Dz\, z^2 
  - \q^{-2}(-\q z^2\,\q^{-2}\Dz + z) = 0
\end{multline*}
by \bref{dz-relations}. \ Similarly,
\begin{multline*}
  F(\Dz\, z - (\q - \q^{-1}) - \q^{-2} z\,\Dz)
  =
  K^{-1}(\Dz)\,F(z) + F(\Dz) z
  - \q^{-2}(K^{-1}(z) \, F\Dz + F(z)\,\Dz)\\
  = \q^2 \Dz - \q\,\Dz^2 z
  - \q^{-2}(-\q^{-2} z\,\q \Dz^2 + \Dz) = 0
\end{multline*}
by \bref{dz-relations} as well.

\subsubsection{}
As noted in the Introduction, the quantum commutativity property,
Eq.~\eqref{equivariance}, is violated for our $\U$-module algebra; for
example, we have
\begin{gather*}
  \sum R^{(2)}(\Dz)\,R^{(1)}(z)
  = \sum_{j=0}^{p - 1}\gamma_j\,z^{j} \Dz^{j}
\end{gather*}
with the nonzero coefficients
\begin{gather*}
  \gamma_j=\sum_{i=\max(j-1, 0)}^{j + p - 2}
  \!\!\!(\q - \q^{-1})^{2 i - j + 1}
  \q^{-i^2  - 4 i - 2 - \frac{1}{2} (j^2 + 3 j)  - i j} {\qbin{i + 1}{j}}^2
  \qfac{i}\qfac{i  + 1 - j}.
\end{gather*}
Yet in the basis of monomials $z^m\Dz^n$, Eq.~\eqref{equivariance}
holds in the cases noted above, which in particular include all pairs
$v=z^i\Dz^j$, $w=z^{m}$ and all pairs $v=\Dz^j$, $w=z^{m}\Dz^{n}$, for
which all the $\Dz^n$ in $wv$ stand to right of the $z^m$.  For
example, with the $R$-matrix in~\eqref{the-R}, we calculate
\begin{gather*}
  R(\Dz\tensor z)
  =\sum_{i=0}^{1}
  \ffrac{(\q-\q^{-1})^i}{\qfac{i}}\,\q^{\frac{i(i-1)}{2} - 2(i-1)^2 }
  E^i\Dz \tensor F^i z
  =\q^{-2}\Dz\tensor z + (\q-\q^{-1}) 1\tensor 1,
\end{gather*}
and therefore the right-hand side of~\eqref{equivariance} evaluates as
\begin{equation*}
  \sum R^{(2)}(z)\,R^{(1)}(\Dz)= \q-\q^{-1} + \q^{-2} z\,\Dz
  =\Dz\,z.
\end{equation*}
In the commutative subalgebras $\oC[z]$ and $\oC[\Dz]$, even simpler,
\begin{gather*}
  R(z\tensor z)
  =\sum_{i=0}^{1}
  \ffrac{(\q-\q^{-1})^i}{\qfac{i}}\,\q^{\frac{i(i-1)}{2} - 2 (i^2 - 1)}
  E^i z\tensor F^i z = \q^2 z\tensor z + (\q-\q^{-1})(-\q)z^2\tensor 1,
\end{gather*}
which makes~\eqref{equivariance} an identity, and similarly for
$R(\Dz\tensor \Dz)$.

\subsection{The quotient $\bCzd$}
We saw in~\bref{derivations} that $z^p$ and $\Dz^p$ are central in
$\Czd$.  The formulas for the $\U$ action also imply that $E z^p = F
z^p = E \Dz^p = F \Dz^p = 0$.  We can therefore take the quotient of
$\Czd$ by relations~\eqref{^p-relations}.  The quotient $\U$-module
algebra is denoted by $\bCzd$ in what follows.

We note that the derivations in~\bref{derivations} do not descent to
$\bCzd$ because, for example, $\mathfrak{d}(z^p)= p(\q - \q^{-1})\,1$.

\begin{small}%
\subsection{The $\U$ action on $\oC[z]/z^p$ in terms of
    $q$-differential operators}
This subsection is a digression not needed in the rest of this paper.

\subsubsection{``Scaling'' operator $\mathscr{E}$}\label{muZ}
The operator
\begin{gather*}
  \mathscr{E} = \ffrac{\Dz\,z-z\,\Dz}{\q - \q^{-1}}
  = 1 - \q^{-1} z \Dz,
\end{gather*}
commutes with $z$ and $\Dz$ as
\begin{equation*}
  \mathscr{E} z^{n} = \q^{-2 n}z^{n} \mathscr{E},\quad
  \mathscr{E} \Dz^{n} = \q^{2 n} \Dz^{n} \mathscr{E}.
\end{equation*}
In what follows, when we speak of the \textit{action} of
$q$-differential operators on $\oC[z]$, it is of course understood
that $\mathscr{E}(z^{n}) = \q^{-2 n}z^{n}$.

We also calculate
\begin{align*}
  \mathscr{E}^n &= 1 + \sum_{i=1}^{n}
  \qbin{n}{i}(-1)^{i}\q^{-n i} z^{i} \Dz^{i}.
  \\
  \intertext{In particular, $\mathscr{E}^p = 1 + z^{p}\Dz^{p}$, and
    hence}
  \mathscr{E}^p &= 1\quad\text{in}\quad\bCzd.
\end{align*}
Therefore, $\mathscr{E}$ is invertible in $\bCzd$.  Moreover, it is
easy to see that in $\bCzd$, the above formula for $\mathscr{E}^n$
extends to negative $n$ as
\begin{equation*}
  \mathscr{E}^n = 1
  + \sum_{i=1}^{p-1}\ffrac{\qint{n-i+1}\dots\qint{n}}{\qfac{i}}
  (-1)^{i}\q^{-n i} z^{i} \Dz^{i},\quad n\in\oZ,
\end{equation*}
which thus gives an explicit representation for $\mathscr{E}^{-1}$, in
particular.

The next lemma shows that, as could be expected, the $E$ and $F$
generators acting on $\oC[z]/z^p$ are (almost) given by multiplication
by~$z$ and by a $q$-derivative.
\begin{lemma}
  The $\U$ action on $\oC[z]/z^p$ is given by the $q$-differential
  operators
  \begin{align*}
    E &= \fffrac{1}{\q - \q^{-1}}\, z\, (1 - \mathscr{E}^{-1}),\\
    K &= \mathscr{E}^{-1},\\
    F &= \fffrac{1}{\q - \q^{-1}}\,\Dz.
  \end{align*}
\end{lemma}
\begin{proof}
  First, by~\bref{muZ}, $E$, $K$, and $F$ are $q$-differential
  operators.  Next, we verify that the right-hand sides of the three
  formulas above act on the $z^{m}$ as desired.  This suffices for the
  proof, but it is actually rather instructive to verify the $\U$
  commutation relations for the above $E$, $K$, and $F$.  For example,
  we have
  \begin{multline*}
    EF - FE = \fffrac{1}{(\q - \q^{-1})^2}\, z (1 - \mathscr{E}^{-1})\,\Dz
    - \fffrac{1}{(\q - \q^{-1})^2}\, \Dz\,z (1 - \mathscr{E}^{-1})\\
    =
    \fffrac{1}{(\q - \q^{-1})^2}\, (1 - \q^{-2}\mathscr{E}^{-1})\,z\,\Dz
    - \fffrac{1}{(\q - \q^{-1})^2}\, \Dz\,z (1 - \mathscr{E}^{-1})
    =\fffrac{\mathscr{E}^{-1} - \mathscr{E}}{\q - \q^{-1}},
  \end{multline*}
  where in the last equality we substitute $z\Dz=\q(1 - \mathscr{E})$
  and $\Dz\,z=\q - \q^{-1}\mathscr{E}$.
\end{proof}
\end{small}

\subsection{Decomposition of \ $\bCzd$}
\ We now decompose the $\,p^2$-dimensional $\,\U$-module $\bCzd$ into
indecomposable representations.

\subsubsection{$\repP^+_1$}
The projective module $\repP^+_1\subset\bCzd$ is identified very
easily.  For $t$ in~\eqref{the-t}, it follows that
\begin{equation*}
  E\pT=z+\q\,z^p\,\Dz^{p-1},\qquad F\pT=\Dz + \q\,z^{p-1}\,\Dz^p.
\end{equation*}
In $\bCzd$, we therefore have the $\repP^+_1$ module realized as shown
in~\eqref{repPp1} (where, again, the horizontal arrows represent the
action of $F$ and $E$ up to nonzero
factors).\enlargethispage{\baselineskip}

\begin{thm}
  As a $\U$-module, $\bCzd$ decomposes as
  \begin{equation*}
    \bCzd=
    \repP^+_1\oplus\repP^+_3\oplus\dots\oplus\repP^+_\nu,
  \end{equation*}
  where $\nu=p$ if $p$ is odd and $p-1$ if $p$ is even.
\end{thm}
\noindent
(We recall that $\dim\repP^+_n=2p$ for $1\leq n\leq p-1$ and
$\dim\repP^+_p=p$.)

\begin{proof}
  The proof is only half legerdemain and the other half calculation,
  somewhat involved at one point; reducing the calculational component
  would be desirable.

  The module $\repP^+_{1}$ is given in~\eqref{repPp1}.  The module
  $\repP^+_{p}$, which occurs in the direct sum in the theorem
  whenever $p=2s+1$ is odd, is the irreducible representation with the
  highest-weight vector
  \begin{equation*}
    t_1(s) = \sum_{i=0}^{s} \q^{i s} \qbin{s + i - 1}{i}
    z^{i + s} \Dz^{i},\quad p=2s+1.
  \end{equation*}
  Calculating with the aid of
  \begin{align*}
    E(z^m \Dz^n)&=\q^{1-n}\qint{n}z^m \Dz^{n-1}
    -\q^{m-2n}\qint{m} z^{m+1} \Dz^n,\\
    F(z^m \Dz^n)&=\q^{1-m}\qint{m} z^{m-1} \Dz^{n}
    -\q^{n-2m}\qint{n}z^m\Dz^{n+1},
  \end{align*}
  we easily verify that $E t_1(s)=0$; it also follows that $F^{p-1}
  t_1(s)\neq 0$; in fact,
  \begin{equation*}
    F^{p-1} t_1(s)=\qfac{p - 1}\sum_{i=0}^{s} \q^{i s}
    \qbin{s + i - 1}{i} z^{i} \Dz^{i + s}.
  \end{equation*}

  As we know from~\cite{[FGST]}, each of the other $\repP^+_{2r+1}$
  modules for $1\leq r\leq\floor{\frac{p-1}{2}}$ has the structure
  (with $r$ omitted from arguments for brevity)
  \begin{equation}\label{proj(2r+1)}
    \xymatrix@=12pt@C=6pt{%
      &&&t_1
      \ar_{E}[dl]
      &
      \kern-10pt\rightleftarrows \ldots \rightleftarrows\kern-10pt
      &t_{2r+1}\ar^F[dr]
      \\
      l_{p-2r-1}
      &
      \kern-10pt\rightleftarrows \ldots \rightleftarrows\kern-10pt
      &l_1\ar_F[dr]
      &&&&r_1\ar[dl]^E
      &
      \kern-10pt\rightleftarrows \ldots \rightleftarrows\kern-10pt
      &r_{p-2r-1}\\
      &&&b_1
      &
      \kern-10pt\rightleftarrows \ldots \rightleftarrows\kern-10pt
      &b_{2r+1}
    }
  \end{equation}
  and our task is now to identify the corresponding elements
  in~$\bCzd$.

  We begin constructing $\repP^+_{2r+1}$ from the bottom, setting
  \begin{equation*}
    b_1=\sum_{i=0}^{p - r - 1} \ffrac{\qfac{r + i - 1}}{\qfac{i}}\,
    \q^{r i} z^{i + r} \Dz^{i},
  \end{equation*}
  which is easily verified to satisfy the relation $E b_1=0$; also,
  $F^{2r} b_1\neq0$\,---\,in fact,
  \begin{equation*}
    F^{2r} b_1= \qfac{2r}\sum_{i=0}^{p - r - 1}
    \ffrac{\qfac{r + i - 1}}{\qfac{i}}\,
    \q^{r i}z^{i} \Dz^{i + r}
  \end{equation*}
  ---\,and $F^{2r+1} b_1 = 0$.  This completely describes the bottom
  $(2r+1)$-dimensional submodule (the irreducible representation
  $\repX^+_{2r+1}$).

  We next seek $l_1$ such that $b_1=F l_1$; obviously,
  $l_1$ is of the general form
  \begin{equation*}
    l_1=\sum_{i=0}^{p - r - 2} \lambda_i \q^{r i} z^{i + r + 1} \Dz^{i}.
  \end{equation*}
  The condition $b_1=F l_1$ is equivalent to the recursion relations
  (we restore $r$ in the argument)
  \begin{gather}\label{recursion}
    \lambda_{i + 1}(r)\qint{i + r + 2}
    - \q^{-2 r - 1}\qint{i}\lambda_{i}(r)
    = \q^{r + i + 1}\ffrac{\qfac{i + r}}{\qfac{i + 1}}.
  \end{gather}
  The problem is made nontrivial by the existence of \textit{two}
  boundary conditions: we must have
  \begin{align}
    \label{bdry-left}
    \lambda_{0}(r) &= \q^{r} \ffrac{\qfac{r - 1}}{\qint{r + 1}}
    \\[-.5\baselineskip]
    \intertext{and}
    \label{bdry-right}
    \lambda_{p - r - 2}(r) &= \q^{2 r}\ffrac{\qfac{r}}{\qint{r + 2}}
  \end{align}
  simultaneously.

  We now solve the recursion starting from the $i=0$ boundary.  The
  problem is thus to find $\lambda_i(r)$ with $i\geq 1$
  from~\eqref{recursion} and~\eqref{bdry-left} and then verify
  that~\eqref{bdry-right} is satisfied.

  The solution is particularly simple for $r=1$, where
  $\lambda_i(1)=\q^2/\qint{3}$ for all $i\geq 1$.  For $r=2$, the
  solution is ``linear in $i$'':
  \begin{gather*}
    \lambda_i(2)=
    {\qbin{5}{2}}^{-1}(\q^3 \qint{i + 4} + \q^4 \qint{i - 1}),
    \quad i\geq 1.
  \end{gather*}
  For $r=3$, it is ``quadratic'' in a similar sense,
  \begin{gather*}
    \lambda_{i}(3) = {\qbin{7}{3}}^{-1}
    \Bigl(
      \q^4\qint{i + 5}\qint{i + 6}
      + \q^5 \qint{i + 5}\qbin{3}{2}\qint{i - 1} + 
      \q^6 \qint{i - 2}\qint{i - 1}
    \Bigr),
    \quad i\geq 1.
  \end{gather*}
  The general solution is given by
  \begin{multline*}
    \lambda_i(r)=
    {\qbin{2 r + 1}{r}}^{-1}
    \biggl(
    \q^{r + 1}\qbin{i + 2r}{r - 1}\qfac{r - 1}+{}\\
    \qquad\qquad{}+
    \sum_{n=2}^{r - 1} \q^{r + n}\qbin{i + 2r + 1 - n}{r - n}
    \qbin{r - 1}{n}\qbin{r}{n - 1}\qfac{r - n - 1}
    \prod_{j=1}^{n - 1} \qint{i - j}+{}\\
    {}+
    \q^{2 r}\prod_{j=1}^{r - 1}\qint{i - j}
    \biggr),
  \end{multline*}
  $i\geq1$.  The first term in the brackets can be included into the
  sum over $n$, by extending it to $n=1$, but we isolated it because
  this is the only term that does not contain the factor $\qint{i-1}$
  and it clearly shows that the solution starts as ${\qbin{2 r +
      1}{r}}^{-1}\q^{r+1}\qint{i+r+2}\dots\qint{i+2r}$ (all the other
  terms are then found relatively easily from the recursion).  The
  boundary condition at $i=p-r-2$ is remarkably simple to verify: only
  one (the~last) term contributes and immediately yields the desired
  result.

  The structure of the general formula may be clarified with a more
  representative example:
  \begin{small}
  \begin{multline*}
    \lambda_i(5)=
    {\qbin{11}{5}}^{-1}
    \biggl(\q^6 \qint{i + 10}\qint{i + 9}\qint{i + 8}\qint{i + 7} +
    \q^7 \qint{i + 9}\qint{i + 8}\qint{i + 7}\qbin{5}{2}\qint{i - 1}\\
    {}+
    \q^8 \qint{i + 8}\qint{i + 7}\ffrac{\qint{4}}{\qint{2}}
    \qbin{5}{2}\qint{i - 2}\qint{i - 1}
    + \q^9 \qint{i + 7}\qbin{5}{2} \qint{i - 3}\qint{i - 2}\qint{i - 1}
    \\
    {}+ \q^{10}\qint{i - 4}\qint{i - 3}\qint{i - 2}\qint{i - 1}
    \biggr).
  \end{multline*}%
  \end{small}%
  This also illustrates the general situation with the boundary
  condition at $i=p-r-2$ (only the last term is nonzero in
  $\lambda_{p-7}(5)$).

  With the $\lambda_i$ and $l_1$ thus found, the other $l_n$ follow by
  the action of~$E$.  

  All the $r_n$ in~\eqref{proj(2r+1)}, starting with $r_1$ such that
  $E r_1= b_{2r+1}$, are found totally similarly (or, with some care,
  obtained from the $l_n$ by interchanging $z$ and $\Dz$).

  The proof is finished with a recourse to the representation theory
  of~$\U$~\cite{[FGST2]}.  For definiteness, we consider the case of
  an odd $p$, $p=2s+1$.  Then what we have established so far is the
  existence of elements shown with black dots in
  Fig.~\ref{fig:stack-of-P},
  \begin{figure}[tbh]
    \centering
    \begin{equation*}
      \xymatrix@R=7pt@C=1pt{
        &{\scriptstyle s+3}\ar@{.}[];[10,0]
        &{\scriptstyle s+2}\ar@{.}[];[10,0]
        &{\scriptstyle s+1}\ar@{.}[];[10,0]
        &{\scriptstyle\;\ s\;\ }\ar@{.}[];[10,0]
        &{\scriptstyle s-1}\ar@{.}[];[10,0]
        &{\scriptstyle s-2}\ar@{.}[];[10,0]
        &{\scriptstyle s-3}\ar@{.}[];[10,0]
        &{\scriptstyle s-4}\ar@{.}[];[10,0]
        &\ldots\vphantom{s}
        \\
        \repP^+_p:& & & &*{\bullet}\ar@{-}@*{[|(3)]}[r]
        &*{\bullet}\ar@{-}@*{[|(3)]}[r]&*{\bullet}\ar@{-}@*{[|(3)]}[r]
        &*{\bullet}\ar@{-}@*{[|(3)]}[r]&*{\bullet}\ar@{-}@*{[|(3)]}[r]&\ldots\\
        & & & & & & & & & \\
        & & & & &*{\circ}\ar[dl]&*{\ast}&*{\ast}&*{\ast}&\dots\\
        \repP^+_{p-2}:& & &*{\bullet}\ar@{-}@*{[|(3)]}[r]
        &*{\bullet}\ar@{-}@*{[|(3)]}[dr]\ar[dr]& & & & & \\
        & & & & &*{\bullet}\ar@{-}@*{[|(3)]}[r]
        &*{\bullet}\ar@{-}@*{[|(3)]}[r]
        &*{\bullet}\ar@{-}@*{[|(3)]}[r]
        &*{\bullet}\ar@{-}@*{[|(3)]}[r]
        &\ldots\\
        & & & & & & & & & \\
        & & & & & &*{\circ}\ar[dl]& & & \\
        \repP^+_{p-4}:& &*{\bullet}\ar@{-}@*{[|(3)]}[r]
        &*{\bullet}\ar@{-}@*{[|(3)]}[r]&*{\bullet}\ar@{-}@*{[|(3)]}[r]
        &*{\bullet}\ar@{-}@*{[|(3)]}[dr]\ar[dr]& & & & \\
        & & & & & &*{\bullet}\ar@{-}@*{[|(3)]}[r]
        &*{\bullet}\ar@{-}@*{[|(3)]}[r]
        &*{\bullet}\ar@{-}@*{[|(3)]}[r]&\ldots
        \\
        & & & & & & & & &
      }
    \end{equation*}
    \caption{Identifying the projective modules in $\bCzd$.}
    \label{fig:stack-of-P}
  \end{figure}
  for the irreducible projective module $\repP^+_p$ and for what is to
  become the projective modules $\repP^+_{p-2}$, $\repP^+_{p-4}$,
  \dots, $\repP^+_{1}$.  To actually show that the black dots do
  complete to the respective projective modules, we establish the
  arrows (maps by~$E$) from some elements shown with open dots (which
  are thus to become the corresponding $t_1$ in~\eqref{proj(2r+1)}).
  The grading indicated in the figure is such that $\deg z=1$ and
  $\deg\Dz=-1$.  In any grade $u>0$, there are $p-u$ linearly
  independent elements in $\bCzd$:
  \begin{equation*}
    z^u,\quad z^{u+1}\Dz,\quad z^{u+2}\Dz^2,
    \quad \dots,\quad z^{p-1}\Dz^{p-1-u}.
  \end{equation*}
  In grade $s$, in particular, there are $p-s=s+1$ elements, and just
  $s+1$ black dots in all of the $\repP^+_{p}$, $\repP^+_{p-2}$,
  \dots, $\repP^+_{1}$.  But in grade $s-1$, there are $s+2$ linearly
  independent elements, only $s+1$ of which have been accounted for by
  the black dots constructed so far.  We let the remaining
  element\,---\,the open dot in grade $s-1$ in
  Fig.~\ref{fig:stack-of-P}\,---\,be temporarily denoted by
  $\circ_{s-1}$.

  Because grade $s$ is exhausted by black dots, $E(\circ_{s-1})$ is
  either zero or a linear combination of the $\bullet$s.  But it is
  elementary to see that there is only one (up to a nonzero factor, of
  course) element in each grade annihilated by~$E$, and in grade $s-1$
  it has already been found: this is the $b_1$ state (the leftmost
  $\bullet$) in~$\repP^+_{p-2}$ (once again, in what is to become
  $\repP^+_{p-2}$ when we finish the proof).  Therefore,
  $E(\circ_{s-1})$ is a linear combination of the $\bullet$s in
  grade~$s$, but we know from~\cite{[FGST2]} that this can only be the
  corresponding element of the~$\repP^+_{p-2}$ module (the reason is
  that this is the only element in this grade that is annihilated
  by~$F$ in a quotient of~$\bCzd$).

  Once the
  \raisebox{8pt}{$\xymatrix@R=6pt@C=4pt{&*{\circ}\ar[dl]\\
      *{\bullet}&}$} arrow from a \textit{single} element in
  grade~$s-1$ is thus established, the rest of the $\repP^+_{p-2}$
  module is completed automatically~\cite{[FGST2]}.  In particular,
  there are the $\ast$s shown in Fig.~\ref{fig:stack-of-P}, and hence
  just one missing $\bCzd$ element in grade~$s-2$, to which we again
  apply the above argument.  Repeating this gives all of the
  projective modules in~\eqref{the-decomposition}.
\end{proof}

\section{Matrix representation}\label{sec:matrices}
\subsection{}
The matrix representation of the basic commutation
relation~\eqref{the-relation} is found quite straightforwardly (it has
many parallels in the $q$-literature, but nevertheless seems to be
new).  Because both $z$ and $\Dz$ are $p$-nilpotent, the matrices
representing them have to be triangular and start with a
next-to-leading diagonal; Eq.~\eqref{the-relation} then fixes the
matrices as in~\eqref{zDz-matrices} (modulo similarity
transformations).  The rest is just a matter of direct verification
(and, of course, a consequence of the fact that $\dim\bCzd=p^2$).

As regards the $\U$ action in the explicit form~\eqref{diagrams}, we
first verify it on the generators, $\Dz$ and $z$ represented as
in~\eqref{zDz-matrices}, and then propagate to $\Mat_p(\oC)$ in
accordance with the $\U$-module algebra property.

It is amusing to see how the $\U$-module algebra property $h(XY)=\sum
h'(X)h''(Y)$ holds for the ordinary matrix multiplication.  For $h=F$,
for example, we have (
for ``bulk'' values of $i$ and $j$)
\begin{multline*}
  \Bigl(\sum F'(X)F''(Y)\Bigr)_{ij}=
  \sum_{k=1}^p\left(K^{-1}(X)\right)_{ik}\left(F(Y)\right)_{kj}
  + \sum_{k=1}^p\left(F(X)\right)_{ik}\left(Y\right)_{kj}\\
  {}=
  \sum_{k=1}^{p-1}
  \q^{k-2i + 1}x_{ik}\qint{k}y_{k+1,j} -
  \sum_{k=1}^{p}\q^{j - 2i}\qint{j - 1}x_{ik} y_{k,j-1}\\
  \shoveright{{}+
    \q^{1-i}\qint{i}x_{i + 1, 1}y_{1,j}
    + \sum_{k=1}^{p-1}\bigl(\q^{1-i}\qint{i}x_{i+1,k+1}
    - \q^{k - 2i + 1}\qint{k}x_{i,k}\bigr)y_{k+1,j}}\\
  {}=
  -
  \sum_{k=1}^{p}\q^{j - 2i}\qint{j - 1}x_{ik} y_{k,j-1}
  + \sum_{k=0}^{p-1}\q^{1-i}\qint{i}x_{i+1,k+1}y_{k+1,j},
\end{multline*}
which is $\left(F(XY)\right)_{ij}$.  The formulas for $E(XY)$ are
equally straightforward.

\subsection{Examples}
\subsubsection{}
As another example of ``matrices as a visual aid,'' we note that the
cointegral $\coint\in\U$ must map any $X\in\Mat_p(\oC)$ into the unit
matrix times a factor; with the cointegral normalized as
in~\cite{[FGST]},
\begin{equation*}
  \coint=\sqrt{\ffrac{p}{2}}\,\ffrac{1}{([p-1]!)^2}\,
  F^{p-1}E^{p-1}\sum_{j=0}^{2p-1}K^j,
\end{equation*}
we actually have
\begin{equation*}
  \coint(X)=\pmb{1}
  \Bigl((-1)^p\sqrt{2p}\sum_{i=1}^{p} \q^{2i - 1} x_{ii}\Bigr).
\end{equation*}

Also, it is easy to see that in the matrix form, the $b_1$ (bottom
left) element of each $\repP^+_{2r+1}$ ($r\geq1$) is the one-diagonal
lower-diagonal matrix
\begin{equation*}
  (b_1(r))_{ij}=\delta_{i, j + r}\,\q^{2 r (j - 1)}\qfac{r - 1}.
\end{equation*}

\subsubsection{}
We choose the ``moderately large'' value $p=4$ for further
illustration.  Then the idea of how the $\U$ generators act on the
matrices is clearly seen from
\begin{footnotesize}
\begin{multline*}
  \nsize{(\q - \q^{-1})E X={}}
  \begin{pmatrix}
    x_{12} & x_{13} & x_{14} & 0 \\
    -x_{11} + x_{22} & \q^2 x_{12} +x_{23} & x_{13}+x_{24} & -\q^2
    x_{14} \\
    -\q^2 x_{21} + x_{32} & -x_{22} + x_{33} & \q^2 x_{23} +x_{34} &
    x_{24} \\
    x_{31}+x_{42} & -\q^2 x_{32} + x_{43} & -x_{33} + x_{44} & \q^2 x_{34}
  \end{pmatrix},
  \\
  \shoveleft{\nsize{(\q - \q^{-1})^2E^2 X={}}}\\*
  \begin{pmatrix}
    x_{13} & x_{14} & 0 & 0 \\
    (\q^2-1) x_{12}+x_{23} & (\q^2+1) x_{13}+x_{24} & (1-\q^2) x_{14} & 0 \\
    \q^2 x_{11} -(\q^2+1) x_{22}+x_{33}
    & -\q^2 x_{12} +(\q^2-1) x_{23}+x_{34}
    & \q^2 x_{13} +(\q^2+1) x_{24} & -\q^2 x_{14} \\
    -\q^2 x_{21} +(1-\q^2) x_{32}+x_{43} & \q^2 x_{22}
    -(\q^2+1) x_{33}+x_{44} & (\q^2-1) x_{34}-\q^2 x_{23} & \q^2 x_{24}
  \end{pmatrix},
  \\
  \nsize{(\q - \q^{-1})^3E^3 X={}}
  \begin{pmatrix}
    x_{14} & 0 & 0 & 0 \\
    \q^2 x_{13} +x_{24} & x_{14} & 0 & 0 \\
    x_{12}-x_{23}+x_{34} & \q^2 x_{24}-x_{13} & x_{14} & 0 \\
    \q^2 x_{11} - \q^2 x_{33} -x_{22}+x_{44} &
    -x_{12}+x_{23}-x_{34} & -\q^2 x_{13} -x_{24} & x_{14}
  \end{pmatrix},
\end{multline*}
{\normalsize and}
\begin{gather*}  
  \nsize{FX={}}
  \begin{pmatrix}
    x_{21} & x_{22}-x_{11} & (-\q^2-1)
    x_{12}+x_{23} & x_{24}-\q^2 x_{13} \\
    (1-\q^2) x_{31} & \q^2 x_{21} +(1-\q^2) x_{32} & (\q^2-1) x_{22}+(1-\q^2)
    x_{33} & (1-\q^2) x_{34}-x_{23} \\
    -\q^2 x_{41} & x_{31}-\q^2 x_{42} & (\q^2+1)
    x_{32}-\q^2 x_{43} & \q^2 x_{33}-\q^2 x_{44} \\
    0 & -\q^2 x_{41} & (1-\q^2) x_{42} & x_{43}
  \end{pmatrix}.
\end{gather*}
\end{footnotesize}

\section{Differential calculus on $\OCzd$}\label{sec:calculus}
We construct a quantum de~Rham complex $(\OCzd,\,d)$ of $\bCzd$ where
the differential $d$ commutes with the $\U$ action.  This requires
introducing a somewhat unusual (compared to the quantum plane
case~\cite{[WZ],[M]}) action of $\U$ on the differentials $d
z\equiv\dz$ and $d\Dz\equiv\dDz$.

\subsection{}
Let $\oC_{\q}[\dz,\dDz]$ be the unital algebra with the relations
\begin{equation}\label{super-plane}
  \begin{gathered}
    \dz^2=0,\quad \dDz^2=0,\\
    \dDz\,\dz=-\q^{-2}\dz\,\dDz.
  \end{gathered}
\end{equation}
On $\Czd\tensor\oC_{\q}[\dz,\dDz]$, we define the differential as
\begin{equation}\label{the-differential}
  d(z)=\dz,\quad d(\Dz)=\dDz,\quad d(\dz)=0,\quad d(\dDz)=0
\end{equation}
(and $d(1)=0$) and set
\begin{equation}\label{all-relations}
  \begin{alignedat}{2}
    \dz\,z&=\q^{-2}z\,\dz,&\qquad \dDz\,\Dz&=\q^{2}\Dz\,\dDz,\\
    \dz\,\Dz&=\q^{2}\Dz\,\dz,&\qquad \dDz\,z&=\q^{-2}z\,\dDz.
  \end{alignedat}
\end{equation}
The first line here immediately implies that
\begin{gather*}
  d(z^m)=\q^{1-m}\qint{m}z^{m-1}\dz,\qquad
  d(\Dz^n)=\q^{n-1}\qint{n}\Dz^{n-1}\dDz.
\end{gather*}
\begin{lemma}
  The algebra on $z$, $\Dz$, $\dz$, and $\dDz$ with relations
  \eqref{the-relation}, \eqref{super-plane}, and~\eqref{all-relations}
  and differential~\eqref{the-differential} is an associative
  differential algebra.
\end{lemma}
\noindent
The proof is by direct verification.\footnote{As regards comparison
  with the more familiar case of the Wess--Zumino differential
  calculus on the quantum plane~\cite{[WZ],[M]}, it may be interesting
  to note that the associativity requires the vanishing of
  \textit{both} coefficients $\nu$ and $\beta$ in the tentative
  relations $\dz\,\Dz = \mu\,\Dz\,\dz + \nu\,z\,\dDz$ and $\dDz\,z =
  \alpha\,z\,\dDz + \beta\,\Dz\,\dz$.  However, similarities with the
  quantum plane, genuine of superficial, come to an end when we
  consider the quantum group action: the formulas in~\bref{qgaction}
  bear little resemblance to the quantum plane case.}

\subsection{}\label{qgaction}We next define a $\U$ action on the above
algebra by setting
\begin{alignat*}{3}
  E\dz &= -\qint{2}z\,\dz,&\qquad
  K\dz &= \q^2 \dz,&\qquad
  F\dz &= 0,\\
  E\dDz &= 0,&
  K\dDz &= \q^{-2} \dDz,&
  F\dDz &= -\q^{2}\qint{2}\Dz\,\dDz.
\end{alignat*}

\begin{lemma}\label{diff-U-module}
  This defines a differential $\;\U$-module algebra.
\end{lemma}
\noindent
The proof amounts to verifying that this action preserves the
two-sided ideal generated
by~\eqref{super-plane}--\eqref{all-relations}.

\subsubsection{}\label{action-formulas}We note simple consequences of
the above formulas:
\begin{align*}
  E^{i}(z^{m}\,\dz) &= (-1)^i \q^{i m + \frac{i(i-1)}{2}}
  \qbin{m + i + 1}{i}\qfac{i}\, z^{m + i}\,\dz,\\
  F^{i}(z^{m}\,\dz) &= \q^{i(1 - m) + \frac{i(i-1)}{2}}
  \qbin{m}{m - i}\qfac{i}\,z^{m - i}\,\dz,\\
  E^{i}(\Dz^{m}\,\dDz) &= \q^{-i(m + 1) + \frac{i(i-1)}{2}}
  \qbin{m}{m - i}\qfac{i}\,\Dz^{m - i}\,\dDz,\\
  F^{i}(\Dz^{m}\,\dDz) &= (-1)^i \q^{i(m + 2) + \frac{i(i-1)}{2}}
  \qbin{m + i + 1}{i}\qfac{i}\,\Dz^{m + i}\,\dDz.
\end{align*}
In particular,
\begin{align*}
  E(z^{m}\,\dz) &= - \q^m \qint{m + 2}\,z^{m + 1}\,\dz,\\
  F(\Dz^{m}\,\dDz) &= - \q^{m + 2} \qint{m + 2}\,\Dz^{m + 1}\,\dDz.
\end{align*}

\subsection{}Because $d(z^p)=0$ and $d(\Dz^p)=0$, it follows that the
differential $\U$-module algebra structure descends to
the quotient by the relations $z^p=0$ and $\Dz^p=0$.  We finally let
$\OCzd$ denote the resulting differential $\U$-module
algebra\,---\,the sought quantum de~Rham complex
\begin{equation*}
  \OCzd=(\bCzd\tensor\oC_{\q}[\dz,\dDz],d)\!\bigm/\!\mathcal{I}\!,
\end{equation*}
where $\mathcal{I}$ is the ideal generated by~\eqref{the-relation},
\eqref{^p-relations},
and~\eqref{super-plane}--\eqref{all-relations}.

As a vector space, $\OCzd$ naturally decomposes into zero-, one- and
two-forms.  In $\Omega^1\bCzd$, the elements $z^{p-1}\,\dz$ and
$\Dz^{p-1}\,\dDz$ are the cohomology of $d$ (the ``cohomology
corners'' of the modules shown in~\eqref{two-modules}).

\section{Conclusions}
As noted above, it is a classic result that (using the modern
nomenclature) the matrix algebra is generated by the generators $x$
and $y$ of a finite quantum plane (with $x^p=y^p=1$) at the
corresponding root of unity~\cite{[W]}; it may be even better known
that the quantum plane carries a quantum-$\SL2$
action~\cite{[WZ],[M]}; and the two facts can of course be combined to
produce a quantum-$\SL2$ action on matrices (cf.~\cite{[DNS],[CS]}).
We construct an action of $\UresSL2$ at $\q=e^{\frac{i\pi}{p}}$ on
$p\times p$ matrices starting not from the quantum plane but from
$q$-differential operators on a ``quantum line''; the explicit
formulas for this action are not altogether unworthy of consideration.

Also, the $\UresSL2$-module algebra constructed here (and most
``invariantly'' described in terms of $q$-differential operators) is
relevant in view of the Kazhdan--Lusztig correspondence between
$\UresSL2$ and the $(p,1)$ logarithmic conformal models.  Previously,
the Kazhdan--Lusztig correspondence in logarithmic conformal field
theories has been observed to hold at the level of representation
theories (of the quantum group and of the chiral algebra) and modular
transformations (on the quantum group center and on generalized
characters of the chiral
algebra)~\cite{[FGST],[FGST2],[FGST3],[FGST-q],[S-q]}.  Our results
show how it can be extended to the level of fields, the key
observation being that the object required on the quantum-group side
is an algebra with ``good'' properties under the action of $\U$ and
with a differential that commutes with this action.

Another possibility to look at the Kazhdan--Lusztig correspondence is
offered just by the $\UresSL2$-module algebra defined on
$\Mat_p(\oC)$: a ``spin chain'' can be defined by placing the algebra
generated by $z$ and $\Dz$ at each site (as we remember, these
generalize free fermions, which indeed occur at $p=2$).  In choosing
the Hamiltonian, an obvious option is to have it related to the
Virasoro generator $L_0$; a suggestive starting point on a finite
lattice is the relation~\cite{[FGST2]}
\begin{equation*}
  e^{2i\pi L_0}=\ribbon,
\end{equation*}
where $\ribbon$ is the ribbon element in $\UresSL2$.  In the matrix
language, the spin chain with the $\UresSL2$-module algebra generated
by $z$ and $\Dz$ at each site is equivalently described just by
letting $\UresSL2$ act on $\Mat_p(\oC)\tensor\Mat_p(\oC)\tensor\dots$,
which may be helpful in practical computations.  (This construction
may have some additional interest because the relevant action is
nonsemisimple (cf.~\cite{[HR],[NRG],[PRZ],[RS]}), but at the same time
the indecomposable representations occurring here are under control
due to the decomposition in~\eqref{the-decomposition}.)  In addition,
it is also interesting to answer several questions ``on the $\bCzd$
side,'' such as where the even-dimensional modules $\repX^+_{2r}$ and
their projective covers $\repP^+_{2r}$ are hiding.

\subsubsection*{Acknowledgments}This paper was supported in part by
the RFBR grant 07-01-00523 and the grant LSS-1615.2008.2. \ I thank
A.~Gainutdinov for the useful comments and P.~Pyatov for remarks on
the literature.

\appendix
\section{OPE algebras and parafermionic statistics}
\label{sec:OPE}
We outline how the parafermionic statistics can be incorporated into
conformal field theory.

\subsection{Background: OPE}
For conformal fields (operators) $A(w)$, $B(w)$, \dots\ defined on the
complex plane, the purpose of the OPE
algebra~\cite{[BBSS],[Th]}\footnote{We proceed in rather down-to-earth
  terms; see~\cite{[Ros]} and the references therein for a much more
  elaborate approach.} is to calculate the expressions (referred to as
OPE poles) $[A,B]_n$ in ``short-distance expansions''
\begin{equation}
  A(z)\,B(w) = \sum_{n\ll\infty}{\frac{[A, B]_n(w)}{(z-w)^n}}
\end{equation}
for any composite operators $A$ and $B$ in terms of the $[~,~]_m$
specified for a set of ``basis'' operators.  (By a composite operator
of any $A(w)$ and $B(w)$, we mean $[A,B]_0(w)$, which is also called
the normal-ordered product and is often written as $AB(w)$ or
$A(w)B(w)$.)  The rules for calculating the OPEs
are~\cite{[BBSS],[Th]}
\begin{gather*}
  [B,A]_n=(-1)^{A B}\sum_{\ell\geq n}\ffrac{(-1)^\ell}{(\ell-n)!}
  \diff^{\ell-n}[A,B]_\ell,\\
  [A, [B, C]_0]_n=(-1)^{A B}[B, [A,C]_n]_0
  + \sum_{\ell=0}^{n-1}\nbin{n-1}{\ell}[[A,B]_{n-\ell},C]_\ell,
\end{gather*}
where in the sign factor $(-1)^{A B}$\,---\,the signature of the Fermi
statistics\,---\,$A$ and $B$ denote the Grassmann parities of the
corresponding operators.\footnote{And $\diff$ is the operator of
  differentiation with respect to the coordinate on the complex plane;
  we use this notation instead of the more common $\partial$ so as not
  to add to the notation overload already existing with~``$z$,'' which
  is now a coordinate on the complex plane along with $w$.}

The first of the above rules allows computing the ``transposed'' OPE
$B(z)\,A(w)$ once the OPE $A(z)B(w)$ is known; the second rule is the
prescription for calculating an OPE with a composite operator
$[B,C]_0$.  There is a third rule stating that $\diff$ acts on the
normal-ordered product $[A,B]_0$ as derivation.  These three rules
(and the simple relation $[\diff A,B]_{n}=-(n-1)[A,B]_{n-1}$) suffice
for the calculation of any OPE of composite operators~\cite{[Th]}.

Each of the two formulas above inevitably contains an inversion of the
operator order (accompanied by a sign factor for fermions); this is
where a generalization to the parafermionic statistics is to be made.

\subsection{Parafermionic OPE}
\label{q-OPE}
We assume that the fields carry a quantum group action and that an
$R$-matrix is given.  As a generalized ``transposition'' OPE rule, we
then postulate
\begin{gather}\label{BA-OPE}
  [B,A]_k=\sum_{\ell\geq k}\ffrac{(-1)^\ell}{(\ell-k)!}
  \diff^{\ell-k}[R^{(2)}(A), R^{(1)}(B)]_\ell,
\end{gather}
where $R^{(2)}$ and $R^{(1)}$ are understood just as
in~\eqref{equivariance} (Sweedler's summation is implied), and where
we assume that all the OPEs in the right-hand side are known.  For the
``composite'' OPE rule, similarly, we set
\begin{gather}\label{ABC-OPE}
  [A, [B, C]_0]_k=[R^{(2)}(B), [R^{(1)}(A),C]_k]_0
  + \sum_{\ell=0}^{k-1}\nbin{k-1}{\ell}[[A,B]_{k-\ell},C]_\ell.
\end{gather}

The consistency of these formulas is not obvious a priori, already
because of the new fields, except $B$ and $A$ themselves, occurring
under the action of the ``right and left coefficients'' of the
$R$-matrix, in $R^{(2)}(B)$ and $R^{(1)}(A)$.  In general, moreover,
whenever a transposition of two fields does not square to the identity
transformation (the situation generally referred to as ``fractional
statistics''), some cuts on the complex plane must be chosen (or a
cover of the complex plane should be specified on which the fields are
defined).  Furthermore, the proposed OPE rules should also be extended
to include possible occurrences of $\log(z-w)$, which we leave for
future work.  But it is interesting to see how the scheme may work for
our $R$-matrix~\eqref{the-R} and ``parafermionic'' fields modeled on
the projective module in~\eqref{repPp1}.

\subsection{The $\UresSL2$ example}
We introduce $p-1$ pairs of conformal fields $\upzeta^m(w)$ and
$\updelta^m(w)$, $m=1,\dots,p-1$, carrying the same $\U$ action as the
$z^m$ and $\Dz^m$ in Sec.~\ref{sec:qline-all}, i.e.,
\begin{alignat*}{2}
  &\begin{aligned}
    E^{i} \upzeta^{m}(w) &= (-1)^i \q^{i m + \frac{i(i - 1)}{2}}
    \qbin{i+m-1}{m-1}\qfac{i}
    \,\upzeta^{m + i}(w),\quad\\
    F^{i} \upzeta^{m}(w) &= \q^{i(1 - m) + \frac{i(i - 1)}{2}}
    \qbin{m}{m-i}\qfac{i}
    \,\upzeta^{m - i}(w),
  \end{aligned}&
  K \upzeta^{m}(w) &=\q^{2m}\upzeta^{m}(w),
  \\
  &\begin{aligned}
    E^{i} \updelta^{m}(w) &= \q^{i(1 - m) + \frac{i(i - 1)}{2}}
    \qbin{m}{m-i}\qfac{i}
    \,\updelta^{m - i}(w),\\
    F^{i} \updelta^{m}(w) &= (-1)^i \q^{i m + \frac{i(i - 1)}{2}}
    \qbin{i+m-1}{m-1}\qfac{i}
    \,\updelta^{m + i}(w),
  \end{aligned}&
  K \updelta^{m}(w) &= \q^{-2m}\updelta^{m}(w),
\end{alignat*}
with $\updelta^0(w)=\upzeta^0(w)=1$ (and, formally,
$\updelta^m(w)=\upzeta^m(w)=0$ for $m<0$ or $m\geq p$).  Here,
$w\in\oC$, which is our ``space--time.''

We also have the derivative of each field, $\diff\upzeta^{m}(w)$ and
$\diff\updelta^{m}(w)$, which we view as space--time $1$-forms, and
hence regard $d$ as a differential.  The differential must commute
with the quantum group action, just as the differential $d$ in
Sec.~\ref{sec:calculus}, which allows the algebraic constructions
involving the differential to be carried over to the fields.

To summarize the notational correspondence between
Secs.~\ref{sec:qline-all}--\ref{sec:calculus} and this Appendix, we
write the dictionary
\begin{alignat}{3}
  \nonumber
  z^m|_{\text{Sec.~\ref{sec:qline-all}}}&
  \leftrightarrow\upzeta^m(w)|_{\text{App}},
  &\quad
  \Dz^m|_{\text{Sec.~\ref{sec:qline-all}}}&
  \leftrightarrow\updelta^m(w)|_{\text{App}},&\quad
  m&=0,\dots,p-1,
  \\
  d(z^m)|_{\text{Sec.~\ref{sec:qline-all}}}&
  \leftrightarrow\diff\upzeta^m(w)|_{\text{App}},
  &\quad
  d(\Dz^m)|_{\text{Sec.~\ref{sec:qline-all}}}&
  \leftrightarrow\diff\updelta^m(w)|_{\text{App}},&\quad
  m&=1,\dots,p-1
  \label{dictionary}
  \\
  \intertext{(we recall that $\upzeta^0(w)=\updelta^0(w)=1$),
    or, using~\eqref{eq:eta},}
  &&\smash[t]{\Dz^{m-1}\,\dDz|_{\text{Sec.~\ref{sec:calculus}}}}
  &\smash[t]{\leftrightarrow\upeta^m(w)|_{\text{App}}},
  &m&=1,\dots,p-1.
  \nonumber
\end{alignat}

\subsubsection{}
Either $E$ or $F$ (depending on the conventions) is to be associated
with the action of a screening operator in conformal field theory
(cf.~\cite{[FGST]}); screenings commute with Virasoro generators and
therefore do not change the conformal weight.  Because we have the
maps $F:\upzeta^1(w)\to 1$ and $E:\updelta^1(w)\to 1$, it follows that
both $\updelta^n(w)$ and $\upzeta^n(w)$ must have conformal weight~$0$
(see~\eqref{fieldsPp1}).

We then fix the basic OPEs of weight-$0$ fields:
\begin{gather*}
  \updelta^m(z)\,\upzeta^n(w)=\qint{m}\delta^{m,n} \log(z-w).
\end{gather*}
Nonlogarithmic OPEs occur when the derivative of either $\upzeta^n(w)$
or $\updelta^n(z)$ is taken:
\begin{gather*}
  \diff\updelta^m(z)\,\upzeta^n(w)=\mfrac{\qint{m}\delta^{m,n}}{z-w},
  \qquad
  \updelta^m(z)\,\diff\upzeta^n(w)=-\mfrac{\qint{m}\delta^{m,n}}{z-w}.
\end{gather*}

\subsubsection{}\label{Lambda-reversed}
As we have noted, fractional-statistics fields generally require cuts
on the complex plane, because taking one of such fields around another
is not an identity transformation.  Therefore, for each ordered pair
of fields $(A,B)$, we must specify whether formula~\eqref{BA-OPE} is
to be used with $R$ or $R^{-1}$.
The rule that we adopt in the current case can be formulated in terms
of diagrams of type~\eqref{fieldsPp1}: we do \textit{not} use the
formulas with the $R$-matrix when both $R^{(1)}$ and $R^{(2)}$ act
toward the socle (the bottom submodule) in~\eqref{fieldsPp1}.

For example, this rule allows rewriting $\upLambda$ with the reversed
normal-ordered products as
\begin{equation}\label{eq:Lambda-reversed}
  \upLambda
  =\sum_{n=1}^{p-1}\ffrac{1}{\qint{n}}\,
  [R^{(2)}(\updelta^n), R^{(1)}(\upzeta^n)]_0
  =\sum_{n=1}^{p-1}\!\sum_{i=0}^{p-1}
  \ffrac{g(i,n)}{\qint{n}}
  [\updelta^{n+i},\upzeta^{n+i}]_0
  =\sum_{n=1}^{p-1}\!\ffrac{\q^{-2n}}{\qint{n}}\,
  [\updelta^n,\upzeta^n]_0,
\end{equation}
where both $R^{(2)}\sim F^i$ and $R^{(1)}\sim E^i$ act ``to the
outside,'' and where we use the temporary notation
\begin{equation*}
  g(i,n)=(\q-\q^{-1})^i
  \q^{\frac{i(i-1)}{2} - i^2 - i - 2 n (i + n)}
  {\qbin{i+n-1}{n-1}}^2\qfac{i}.
\end{equation*}

The same strategy yields the transposed OPE
$\upzeta^n(z)\,\diff\updelta^m(w)$:
\begin{equation*}
  [\upzeta^m,\diff\updelta^n]_1
  =-[R^{(2)}(\diff\updelta^n), R^{(1)}(\upzeta^m)]_1
  = -\delta^{m,n}
  \sum_{i=0}^{p-1}
  g(i,n)
  \qint{n+i}
  = -\delta^{m,n}\q^{2n}\qint{n},
\end{equation*}
or, in a human-friendly form,
\begin{gather*}
  \upzeta^m(z)\,\diff\updelta^n(w)
  =-\mfrac{\delta^{m,n}\q^{2n}\qint{n}}{z-w}.
\end{gather*}
Thus, the effect of the $R$-matrix reduces in these cases to the phase
factor $\q^{2n}=e^{\frac{2i\pi n}{p}}$ occurring under transposition.

\subsubsection{}As a further example, we use the elementary OPEs just
obtained to calculate
\begin{multline*}
  [\diff\upzeta^m,\upLambda]_1
  =\sum_{n=1}^{p-1}\ffrac{1}{\qint{n}}
  [R^{(2)}(\upzeta^n), [R^{(1)}(\diff\upzeta^m),\updelta^n]_1]_0
  \\
  =
  \sum_{i=0}^{p-1}
  (\q-\q^{-1})^i
  (-1)^i
  \q^{\frac{i(i-1)}{2}  + 2 m (i + m)}
  \qbin{m+i}{m}
  \qbin{i+m-1}{m-1}\qfac{i}\q^{2(m+i)}
  \upzeta^{m}
  =\upzeta^{m}.
\end{multline*}
It then follows that $[\upLambda,\diff\upzeta^m]_1=
-[R^{(2)}(\diff\upzeta^m), R^{(1)}(\upLambda)]_1 = -[\diff\upzeta^m,
\upLambda]_1 = -\upzeta^{m}$ because only the $i=0$ term in the
$R$-matrix contributes to $[~,~]_1$.

Next, trying to directly apply~\eqref{ABC-OPE} to calculate
$[\diff\updelta^m,\upLambda]_1$ as
\begin{equation*}
  [\diff\updelta^m,\upLambda]_1
  =\sum_{n=1}^{p-1}\ffrac{1}{\qint{n}}\Bigl(
  [R^{(2)}(\upzeta^n), [R^{(1)}(\diff\updelta^m),\updelta^n]_1]_0
  + [[\diff\updelta^m,\upzeta^n]_{1},\updelta^n]_0\Bigr),
\end{equation*}
we encounter the forbidden arrangement of maps by the left and right
$R$-matrix coefficients; anticipating the result, we claim that this
is irrelevant in this case (essentially because $\diff$ in
$\diff\updelta^m$ annihilates the submodule spanned by unity), but it is
instructive to avoid the forbidden arrangement by using the
``reversed'' $\upLambda$ in~\eqref{eq:Lambda-reversed}:
\begin{multline*}
  [\diff\updelta^m,\upLambda]_1
  =\sum_{n=1}^{p-1}\ffrac{\q^{-2n}}{\qint{n}}
  [\diff\updelta^m,[\updelta^n,\upzeta^n]_0]_1
  =\sum_{n=1}^{p-1}\ffrac{\q^{-2n}}{\qint{n}}
  [R^{(2)}\updelta^n,[R^{(1)}\diff\updelta^m,\upzeta^n]_1]_0
  \\
  =\sum_{i=0}^{m-1}(\q-\q^{-1})^i \q^{\frac{i(i-1)}{2} + 2m(m-i)}(-1)^i
  \qbin{m-1}{m-i-1}\qbin{m}{m-i}\qfac{i}\updelta^m=\updelta^m.
\end{multline*}
It also follows that $[\upLambda,\diff\updelta^m]_1=-\updelta^m$.

\subsubsection{A ``parafermionic'' $\upzeta\upeta$ system}
Returning to the OPEs in~\bref{Lambda-reversed}, we represent the
derivative of $\updelta^n(w)$ as in~\eqref{eq:eta}.  Then the fields
$\upzeta^n(w)$ and $\upeta^n(w)$, whose OPEs are given by
\begin{equation*}
  \upeta^m(z)\,\upzeta^n(w)=\mfrac{\delta^{m,n}\q^{1-n}}{z-w},
  \qquad
  \upzeta^m(z)\,\upeta^n(w)=-\mfrac{\delta^{m,n}\q^{n+1}}{z-w},
\end{equation*}
make up a $(p-1)$-component ``parafermionic'' first-order system; it
generalizes the free fermions, which are indeed recovered for $p=2$,
when also $m=n=1$ (and $\q=i$).  The behavior of the $\upeta^n(w)$
under the $\U$ action is given by the formulas
in~\bref{action-formulas}, in accordance with the dictionary
in~\eqref{dictionary}.  

Similarly to the case of free fermions, we have the weight-$1$ field
(a current) $\JJ=
\sum_{n=1}^{p-1}\q^{n-1}[\upzeta^n,\upeta^n]_0$.
From~\eqref{two-modules}, we conclude that it participates in the
diagram
\begin{equation}\label{qmark}
  \xymatrix@=12pt@C=6pt{
    \JJ(w)=\displaystyle
    \smash{\sum_{n=1}^{p-1}}
    \q^{n-1}\upzeta^n\upeta^n(w)
    \ar^(.6){F}[dr]
    &&&&?\ar_E[dl]
    \\  
    {}
    &\upeta^1(w)
    &\kern-10pt\rightleftarrows \ldots \rightleftarrows\kern-10pt
    &    
    \upeta^{p-1}(w),
  }
\end{equation}
where it remains to identify the ``cohomology corner'' \textit{in
  terms of fields} (we do not have an $\upeta^p(w)$,
see~\eqref{dictionary}).

The ``corner'' must be a field of the same conformal weight as the
current $\JJ(w)$, but must not be a bilinear combination of the
$\upzeta^m(w)$ and $\upeta^m(w)$.  It is naturally provided by the
setting in~\cite{[FHST]}, where the chiral algebra $W(p)$ and its
representation spaces are defined as the kernel of the ``short''
screening operator $S_-$, whereas the ``long'' screening $S_+$ acts on
the fields.  The action of a screening $S$ amounts to taking the
first-order pole in the OPE with the screening current $s(w)$, which
is often expressed as
\begin{gather*}
  S_\pm=\oint s_\pm(w)
\end{gather*}
(with a contour integration over $w$ implied).  In the standard
realization in terms of a free bosonic field $\phi(w)$, we have
$s_+(w)=e^{\sqrt{2p}\,\phi(w)}$ and
$s_-(w)=e^{-\sqrt{2/p}\,\phi(w)}$.
\ With $E\in\U$ identified with the screening operator $S_-$, we now
rescale the grading used in Fig.~\ref{fig:stack-of-P} as follows:
$\JJ(w)$ in~\eqref{qmark} is assigned degree~$0$ and each $F$ arrow
increases the degree by $\sqrt{2/p}$.  Then the question mark
in~\eqref{qmark} has the degree~$\sqrt{2p}$, and therefore \textit{the
  cohomology corner is filled with the screening current~$s_+(w)$}. \
We thus obtain~\eqref{dfieldsPp1}.

A field realization of the other module in~\eqref{two-modules}
requires taking a ``dual'' picture, in terms of the first-order
``parafermionic'' system comprised by the $\diff\upzeta^m(w)$ and
$\updelta^m(w)$ and the $\JJ(w)$ current used to construct the
screening.\footnote{None of these free-field systems, as is well known
  from the $p=2$ example, allows constructing ``logarithmic'' modules
  of the Virasoro or triplet algebra, i.e., indecomposable modules
  where $L_0$ is not diagonalizable.  Logarithmic modules require an
  integration, such as $\diff^{-1}\upeta^n(w)$, leading to the
  $\upzeta^n$, $\updelta^n$ fields.  A remarkable trace of this
  integration may already be observed at the algebraic level
  in~\eqref{the-t}\,---\,the $q$-integers in the denominator and an
  ``integration constant'' $\alpha 1$.}

\subsubsection{}
For the current $\JJ(w)$, the rules in~\bref{q-OPE} lead to the
standard OPE $[\upeta^m,\JJ]_1=\upeta^m$.  Transposing, we then find
$[\JJ,\upeta^m]_1=-[R^{(2)}(\upeta^m), R^{(1)}(\JJ)]_1 =-[\upeta^m,
\JJ]_1=-\upeta^m$ because only the $i=0$ term in $R$-matrix
\eqref{the-R} contributes.  Although $\JJ(w)$ is not a $\U$-invariant,
it behaves like one in a number of OPEs.

We next calculate the first (and the only) pole in the OPE
$\upzeta^m(z)\,\JJ(w)$:
\begin{multline*}
  [\upzeta^m, \JJ]_1
  =\sum_{n=1}^{p-1}\q^{n-1}\,
  [R^{(2)}(\upzeta^n), [R^{(1)}(\upzeta^m),\upeta^n]_1]_0
  \\
  =-
  \sum_{i=0}^{p-1}
  (\q-\q^{-1})^i
  \q^{\frac{i(i-1)}{2} + 2(m+i)m}
  (-1)^i
  \qbin{i+m-1}{m-1}
  \qbin{m+i}{m}\qfac{i}
  \q^{2(m+i)}\,\upzeta^{m}
  =-\upzeta^{m}.
\end{multline*}
It now readily follows that $[\JJ,\upzeta^m]_1=\upzeta^{m}$.

An instructive calculation is that of the $\JJ(z)\,\JJ(w)$ OPE:
\begin{equation*}
  [\JJ,\JJ]_2
  =\sum_{n=1}^{p-1} \q^{n-1}
  [[\JJ,\upzeta^n]_1,\upeta^n]_1
  =\sum_{n=1}^{p-1} \q^{n-1}[\upzeta^n,\upeta^n]_1
  =-\sum_{n=1}^{p-1}\q^{2n} = 1.
\end{equation*}
Thus, although $\JJ(w)$ is a sum of the $p-1$ terms
$\q^{n-1}\,\upzeta^n\upeta^n(w)$, it does \textit{not} show the factor
$p-1$ in the $\JJ(z)\,\JJ(w)$ OPE.

Naturally, just the same is observed in the ``dual'' description, in
terms of another first-order system, with the current $\II$
in~\eqref{eq:II}.  With the OPEs $[\updelta^m, \II]_1 =-\updelta^m$
and $[\diff\upzeta^m, \II]_1 =\diff\upzeta^{m}$ (where in the last
formula the calculation is very much that for
$[\diff\upzeta^m,\upLambda]_1$),
it follows that $[\II,\II]_2 =-\sum_{n=1}^{p-1}\q^{2n}=1$, just as for
the $\JJ$ current.

\subsubsection{}
The same ``summation to minus unity'' occurs for the simplest
energy--momentum tensor, the normal ordered product
\begin{equation*}
  \TT=
  \sum_{n=1}^{p-1}
  \ffrac{1}{\qint{n}}\,[\diff\upzeta^n,\diff\updelta^n]_0
  =\sum_{n=1}^{p-1}
  \q^{n-1}\,[\diff\upzeta^n,\upeta^n]_0.
\end{equation*}
\textit{It is a $\U$ invariant}, which reduces the OPE calculations to
the standard, except at the last step in calculating half the central
charge:
\begin{equation*}
  [\TT,\TT]_4
  =\sum_{n=1}^{p-1}\q^{n-1}\Bigl(
  3[\diff\upzeta^n,\upeta^n]_2
  + [\diff^2\upzeta^n,\upeta^n]_3
  \Bigr)
  =(3-2)\sum_{n=1}^{p-1}\q^{2n}=-1
\end{equation*}
and, similarly,
\begin{equation*}
  [\TT,\JJ]_3=-1.
\end{equation*}

The energy-momentum tensor can of course be ``improved'' by the
derivative of a current.  The ``$\JJ$-improved'' energy--momentum
tensor
\begin{equation*}
  \tilde\TT=\TT - \beta\diff\JJ
\end{equation*}
has the central charge $-2-12\beta^2+12\beta$, which coincides with
the one of the $(p,1)$ model for
\begin{equation*}
  \beta=\bigl(1 + \ffrac{1}{\sqrt{2p}}\bigr)
  \bigl(1-\sqrt{\fffrac{p}{2}}\bigr).
\end{equation*}

\parindent0pt

\end{document}